\documentclass{article}
\usepackage{booktabs} 
\usepackage[ruled]{algorithm2e} 
\usepackage{amsthm}
\usepackage{fullpage}
\usepackage{amsfonts}
\usepackage{enumitem}
\usepackage{subcaption}
\usepackage{mathtools}
\usepackage{natbib}
\usepackage{xcolor}
\usepackage[hidelinks]{hyperref}
\usepackage[short]{optidef}
\usepackage{tikz}
\usepackage[labelfont=bf]{caption}
\usetikzlibrary{patterns,decorations.pathreplacing}

\SetAlFnt{\small}
\SetAlCapFnt{\small}
\SetAlCapNameFnt{\small}
\SetAlCapHSkip{0pt}
\IncMargin{-\parindent}

\newcommand{\bbR}{\mathbb{R}}
\newcommand{\bbN}{\mathbb{N}}
\newcommand{\calH}{\mathcal{H}}
\newcommand{\poly}{\mathrm{poly}}

\newtheorem{theorem}{Theorem}
\newtheorem{lemma}[theorem]{Lemma}
\newtheorem{corollary}[theorem]{Corollary}
\newtheorem*{theorem*}{Theorem}

\theoremstyle{definition}
\newtheorem{definition}[theorem]{Definition}

\title{Cardinal-Utility Matching Markets:\\
The Quest for Envy-Freeness,\\ Pareto-Optimality, and Efficient Computability}

\author{Thorben Tr\"obst \and Vijay V.\ Vazirani}

\date{\today}

\begin{document}

\maketitle

\begin{abstract}
    Unlike ordinal-utility matching markets, which are well-developed from the viewpoint of both
    theory and practice, recent insights from a computer science perspective have left
    cardinal-utility matching markets in a state of flux.
    The celebrated pricing-based mechanism for one-sided cardinal-utility matching markets due to
    \cite{HZ79}(HZ), which had long eluded efficient algorithms, was finally
    shown to be intractable; the problem of computing an approximate equilibrium is PPAD-complete
    (\cite{VY21,CCPY22}).

    This led us to ask the question: is there an alternative, polynomial time, mechanism for
    one-sided cardinal-utility matching markets which achieves the desirable properties of HZ, i.e.\
    (ex-ante) envy-freeness (EF) and Pareto-optimality (PO)?

    We show that the problem of finding an EF+PO lottery in a one-sided cardinal-utility matching
    market is by itself already PPAD-complete. However, a $(2 + \epsilon)$-approximately envy-free
    and (exactly) Pareto-optimal lottery can be found in polynomial time using the
    Nash-bargaining-based mechanism of \cite{HV21}. Moreover, the mechanism is also $(2 +
    \epsilon)$-approximately incentive compatible.

    We also present several results on two-sided cardinal-utility matching markets, including
    non-existence of EF+PO lotteries as well as existence of justified-envy-free and weak
    Pareto-optimal lotteries.
\end{abstract}

\section{Introduction}

Since the invention of the internet and the rise in mobile computing, one-sided and two-sided
matching markets have become an important part of the economy, driving innovation in the field of
matching-based market design; see \cite{EIV23} for a detailed overview.
Prominent commercial examples of such markets include Google's AdWords market (matching advertisers
to user search queries), ride-hailing services (Uber, Lyft), food delivery (Doordash, Uber Eats),
freelancing (Taskrabbit, Upwork), and vacation rentals (Airbnb, VRBO).

Additionally, there are also many useful matching markets \emph{without} monetary transfers which
are the focus of this paper.
Markets of this kind arise in a variety of scenarios where payments are impractical or even immoral.
Well-known examples include the kidney donor market (matching kidney donors to transplant
recipients), the National Resident Matching Program (matching doctors to hospitals), school choice
(matching students to primary or secondary schools), and course allocation (matching students to
courses).

In what we call a \emph{one-sided matching market}, we are simply given a set $G$ of \emph{goods}
and a set $A$ of \emph{agents} with $|A| = |G|$.
Agents have preferences over the goods and each agent is to be assigned exactly one good.
The goal is to design a mechanism that takes in the preferences of the agents and finds a perfect
matching such that certain desirable properties---most notably efficiency, fairness, and incentive
compatibility---are achieved.

Since each agent must get exactly one good, it is generally not possible to assign indivisible goods
in a fair way.
For this reason, it is customary to study probability distributions---or \emph{lotteries}---over
matchings instead and to consider notions of \emph{ex-ante} fairness, efficiency, etc.

Such matching markets can be broadly distinguished into two classes based on the kind of preferences
the agents have over the goods.
In an \emph{ordinal-utility} matching market, each agent $i$ represents their preferences via a
total (or sometimes partial) order $\prec_i \subseteq G^2$, whereas in a \emph{cardinal-utility}
matching market, each agent $i$ has a vector $(u_{i j})_{j \in G}$ of rational, non-negative
utilities instead.

Ordinal preferences are more restrictive which makes them easier to elicit and also makes it easier
to design incentive compatible mechanisms for them.
Several well-known mechanisms such as Probabilistic Serial (\cite{BM01}) and Random
Priority\footnote{also known as Random Serial Dictatorship} (\cite{AT98}) have been developed in the
ordinal setting.

A significant limitation of ordinal preferences is their inability to express the \emph{magnitude}
to which an agent likes a certain good and this can lead to a significant loss of efficiency.
For a simple example, consider a course allocation problem in which students rank courses that they
wish to take.
Student A ranks course 1 as 10/10 and course 2 as 1/10 whereas student B ranks course 1 as 10/10 and
course 2 as 9/10.
If there is only one remaining slot in each course, it makes sense to assign student A to course 1
and student B to course 2 in order to greatly increase the combined social welfare of the
students.
However, from the perspective of an ordinal mechanism, both students have identical preferences and
the only reasonable solution is to randomly assign the students to the courses.

In the previous example, we can improve the social welfare by using cardinal utilities, but this
would come at the cost of envy-freeness: student A envies student B.
\cite{ILWM17} provide an even stronger example of the power of cardinal utilities.
There are instances with $n$ agents and goods in which all agents have identical ordinal preferences
and, as such, the goods would be evenly distributed among the agents under any reasonable ordinal
mechanism.
But, remarkably, another allocation exists which is envy-free and improves the cardinal utilities of
\emph{every} agent by a factor of $\theta(\log n)$.

Accordingly, cardinal utilities have found several applications in the real world.
The popular fair division website Spliddit\footnote{\url{www.spliddit.org}} uses
cardinal utilities for all of its tools to allow users to split goods, chores, rent, etc.\ among
themselves in a fair and efficient way (\cite{CKMPSW19}).
Several academic conferences within theoretical computer science crucially rely on a
cardinal-utility matching market for their review process by asking their committee members to rank
papers on a scale (e.g.\ -20 to 20) in order to assign papers to committee members for review.
Cognomos\footnote{\url{www.cognomos.com}} successfully incorporates cardinal utilities to find
course allocations for universities in a fair and efficient way (\cite{BCKO17}).

When it comes to cardinal utilities, the most notable mechanism is that of \cite{HZ79}, which is
based on a pricing approach.
It results in allocations which are (ex-ante) envy-free (EF) and Pareto-optimal (PO), thus
satisfying the most common notions of fairness and efficiency respectively.
Later, \cite{HMPY18} showed that the HZ mechanism is also incentive compatible in the
large.
Note no mechanism is EF, PO, and also incentive compatible in the traditional sense, i.e.\ without
the ``in the large'' restriction (\cite{Z90}). We review the HZ mechanism in
Section~\ref{sec:one_sided}.

However, a core issue with the approach is that \emph{computing} an HZ equilibrium is hard in theory
and in practice.
\cite{VY21} recently showed that there are instances in which every
HZ equilibrium is irrational.
They also showed that the problem of finding an approximate HZ equilibrium is in PPAD and
conjectured that it is PPAD-complete.
This conjecture was confirmed by \cite{CCPY22} who proved the corresponding hardness
result.

This motivates the search for alternative mechanisms that can achieve some or all of the desirable
properties of HZ while being implementable in polynomial time.
In particular, we pose the question: is it possible to compute an envy-free and Pareto-optimal
lottery in a one-sided cardinal-utility matching market in polynomial time?

Beyond one-sided matching markets and HZ, there are also \emph{two-sided matching markets} in which
we are matching agents to other agents.
For example, this would apply to the aforementioned school choice market if we allow the schools to
also have preferences over students, e.g.\ via aptitude scores.
Except for a few highly restricted special cases (\cite{BM04,RSU05}), it was not known whether
envy-free and Pareto-optimal lotteries even exist in this setting.

\subsection{Our Contributions}

\subsubsection*{One-Sided Matching Markets}

Our most significant contribution is that we resolve the question regarding the complexity of
finding EF+PO allocations by showing that the problem is PPAD-hard.
Together with a recent result by \cite{CHR23} this shows that this problem is already PPAD-complete.

\begin{theorem*}[Section \ref{sec:ppad_hardness}]
    The problem of finding an EF+PO allocation in a one-sided cardinal-utility matching market is PPAD-hard.
\end{theorem*}

Our proof works through a polynomial reduction of approximate HZ to the problem of finding EF+PO
allocations which is inspired by the fact that HZ allocations and EF+PO allocations coincide in
certain \emph{continuum markets} involving infinitely many agents and goods (\cite{AS16}).
The key idea is to take an HZ instance and add agents and goods so as to approximate such a
continuum market without perturbing the HZ equilibria in the instance too much.
However, the fact that this yields a working reduction is nonetheless surprising and requires
additional ideas since it was already known that EF+PO allocations need not be approximately HZ,
even in markets that converge to a continuum market in the limit (\cite{MP16}).

Along the way, we will also provide a simple polyhedral proof that there are always rational EF+PO
allocations.
This of course follows from the PPAD membership proof given by \cite{CHR23} but our argument does
not rely on the substantial amount of machinery inherent to proving PPAD membership.

Lastly, we show that the Nash-bargaining-based mechanisms for matching markets introduced by
\cite{HV21} satisfy an approximate notion of envy-freeness and incentive compatibility.

\begin{theorem*}[Section \ref{sec:one_sided_nash}]
    The Nash bargaining solution for one-sided cardinal utility matching markets is 2-envy-free and
    2-incentive compatible.
\end{theorem*}

Together with the algorithms given by \cite{PTV21}, this makes for a polynomial time mechanism that
is $(2 + \epsilon)$-envy-free, $(2 + \epsilon)$-incentive compatible and Pareto optimal.
We remark that HZ is $(1 + \epsilon)$-incentive compatible, however this requires the assumption of
a ``large market'' in which every agent has many copies (\cite{HMPY18}).
Our results establish the Nash-bargaining-based mechanism as a more practical alternative mechanism
for one-sided cardinal-utility matching markets.
We believe that this is an important development in the growing field of matching-based market
design.

\subsubsection*{Two-Sided Matching Markets}

For two-sided markets, the only cases in which it was previously known that EF+PO allocations exist
are when the utilities are in $\{0, 1\}$ and symmetric, i.e.\ each pair of agents either finds their
match mutually agreeable or mutually disagreeable.
In Section~\ref{sec:non_existence}, we provide counterexamples that show that both of these
conditions are necessary: if agents have $\{0, 1, 2\}$ utilities or asymmetric $\{0, 1\}$ utilities,
then EF+PO allocations may not exist.

Given this non-existence result, we give a notion of \emph{justified envy-freeness} (JEF) which is
related to---but to the best of our knowledge different from---notions of fractional stability
from the stable matching literature.
We show existence of rational JEF + weak PO allocations via a limiting argument, an equilibrium
notion introduced by \cite{M16}, and similar polyhedral techniques as we used for
one-sided markets.

\begin{theorem*}[Section \ref{sec:justified_envy_freeness}]
    In any two-sided cardinal-utility matching market, a rational JEF + weak PO allocation always exists.
\end{theorem*}

The Nash-bargaining-based approach by \cite{HV21} and the efficient algorithms by \cite{PTV21}
also extend to two-sided markets.
However, in Section~\ref{sec:two_sided_nash}, we give a counterexample to show that in a Nash
bargaining solution, agents can have $\theta(n)$-factor justified envy toward other agents.

\subsection{Related Work}

Our work builds on the existing literature surrounding the Hylland Zeckhauser mechanism
(\cite{HZ79}) and the complexity of computing HZ equilibria.
\cite{AJT17} give an algorithm to compute HZ equilibria which is based on the algebraic
cell decomposition technique (\cite{Basu1995}).
However, this algorithm needs to enumerate at least $n^{5 n^2}$ cells and is thus highly
impractical even for small values of $n$.
\cite{VY21} give a polynomial time algorithm that computes HZ equilibria for
$\{0, 1\}$ utilities.
They also show FIXP membership for the problem of computing HZ equilibria and PPAD membership for
the problem of computing approximate HZ equilibria.
\cite{CCPY22} show the corresponding PPAD-hardness result, though it remains open
whether finding an exact equilibrium is FIXP-hard.

The notion of envy-freeness comes from fair division where it was originally introduced in the
context of dividing a single resource amongst the agents (\cite{F67,V74}), a problem that is now
referred to as the cake cutting problem.
It also features prominently in the literature on fair division of indivisible goods.
Since it is generally impossible to achieve envy-freeness with indivisible goods, relaxations such
as envy-freeness up to one good (EF1) (\cite{B11}) or envy-freeness up to any good (EFX)
(\cite{CKMPSW19}) are studied instead.

\cite{CT21} recently showed that envy-free and Pareto-optimal lotteries exist in a
large class of (one-sided) fair division problems that in particular includes our setting.
Building on this and recent results by \cite{FHHH23}, \cite{CHR23} established PPAD membership
for this class of problems, though they leave open the question of showing PPAD-hardness which we
resolve here.
They also show that maximizing social welfare over the set of envy-free lotteries is NP-hard, though
their construction relies on a more general problem than the matching markets discussed in this
paper.

Markets with a continuum of agents were introduced by \cite{A64}.
\cite{Z92} showed that in such continuum markets and under locally non-satiating utilities,
envy-free and Pareto-optimal allocations coincide with allocations that come from competitive
equilibria with equal incomes.
In matching markets, the local non-satiation condition is not satisfied, but \cite{AS16} show a
similar equivalence for HZ.
However, \cite{MP16} show that this holds only in the limit: for ``large`` markets, EF+PO
allocations may not be supported by competitive equilibria from approximately equal incomes.

In order to deal with the intractability of HZ, \cite{HV21} recently proposed an alternative,
Nash-bargaining-based mechanism for matching markets.
Their approach works for one-sided and two-sided settings with both linear and non-linear utilities.
Importantly, they show that their Nash bargaining solutions can be computed very efficiently in
practice even on instances with thousands of agents.
The idea of operating markets via Nash bargaining instead of pricing goes back to 
\cite{V12} who used this approach for the linear Arrow Debreu market.
We will introduce the Nash-bargaining-based mechanism in more detail in
Section~\ref{sec:one_sided_nash}.

\cite{PTV21} give algorithms for Nash bargaining in matching markets based on
multiplicative weights update and conditional gradient descent which are efficient in practice and
provide provable running times bounded by $\poly(n, 1 / \epsilon)$.
\cite{AB20} show a reduction from HZ to Nash bargaining in the setting with $\{0, 1\}$ utilities.
They also note that Nash bargaining is not envy-free in general, though, as we will show, it is so
in an approximate sense.

For two-sided markets, there have been several attempts at extending the equilibrium notion of Hylland
and Zeckhauser.
\cite{M16} as well as \cite{EMZ21} introduce equilibrium notions.
In both cases, personalized prices are required, i.e.\ each agent on one side sees a potentially
different set of prices for all other agents on the other side.
In Manjunath's equilibrium we will see that we can still get a kind of justified envy-freeness.

Restricted to symmetric $\{0, 1\}$ utilities, HZ-like equilibria do exist as shown by \cite{BM04}
for bipartite markets and \cite{RSU05} for non-bipartite markets.
A polynomial time algorithm to compute such equilibria and therefore EF+PO allocations was later
given by \cite{LLHT14}.

Beyond this, two-sided markets have been mostly studied under ordinal preferences where stable
matching, as introduced by \cite{GS62}, is the dominant solution concept.
A notable exception is the work by \cite{CFKV19} who study the problem of finding a fractional
stable matching under cardinal utilities that (approximately) maximizes social welfare.

\section{One-Sided Matching Markets}\label{sec:one_sided}

In a one-sided matching market we are given a set $A$ of \emph{agents} and a set $G$ of
\emph{goods}.
We assume that $|A| = |G| = n$ since our goal is to assign exactly one good to each agent (a perfect
matching) in a way that satisfies certain desirable properties.
In this paper we will focus on \emph{cardinal preferences}, that is each agent $i \in A$ has
non-negative utilities $(u_{i j})_{j \in G}$ for every good.
The most notable result in the study of cardinal matching markets is the celebrated
Hylland-Zeckhauser mechanism (\cite{HZ79}) which we will briefly review here.

Hylland and Zeckhauser note that under cardinal utilities it is possible to reduce the case
of \emph{indivisible} goods to the case of \emph{divisible} goods via the Birkhoff-von-Neumann
theorem.

\begin{theorem}[\cite{B46}, \cite{N53}]
    Given a fractional perfect matching $(x_{i j})_{i \in A, j \in G}$, there are $O(n^2)$ integral
    perfect matchings $y^{(1)}, \ldots, y^{(l)}$ and non-negative coefficients $\lambda_1, \ldots,
    \lambda^l$ such that $\sum_{k = 1}^l \lambda_k = 1$ and $x = \sum_{k = 1}^l \lambda_k y^{(k)}$.
    Moreover, both the matchings and the coefficients can be found in polynomial time.
\end{theorem}

This means that if we find a fractional perfect matching $x$, we can simply decompose it into a
convex combination of integral perfect matchings and run a \emph{lottery} over these matchings.
The expected utility of agent $i$ is then $u_i \cdot x_i = \sum_{j \in G} u_{i j} x_{i j}$.
Hence, we will---as Hylland and Zeckhauser did---concern ourselves primarily with fractional
perfect matchings and linear utilities (which we will also just refer to as \emph{allocations}) and
their properties.

Given that we are now in the case of divisible goods with linear utilities, the key insight of
Hylland and Zeckhauser is to employ the power of pricing by implementing a \emph{pseudomarket}.
Each agent is given one unit of fake currency and we determine prices on the goods as well as an
allocation that together form a market equilibrium.

\begin{definition}
    A fractional assignment of goods to agents $(x_{i j})_{i \in A, j \in G}$ and non-negative
    prices $(p_j)_{j \in G}$ form an \emph{HZ equilibrium} if and only if:
    \begin{itemize}
        \item $x$ is a fractional perfect matching,
        \item each agent $i$ spends at most their budget, i.e.\ $\sum_{j \in G} p_j x_{i j} \leq 1$,
            and
        \item each agent $i$ gets a \emph{cheapest\footnote{The condition that the bundle be
            cheapest among all utility-maximizing bundles is a technical sublety. It is required in
            order to ensure that the resulting equilibrium allocations are Pareto-optimal.} utility-maximizing bundle}, i.e.
            \begin{align*}
                u_i \cdot x_i &= \max \left\{ u_i \cdot y \,\middle|\, \sum_{j \in G} y_j = 1, p
                \cdot y \leq 1
                \right\}, \\
                p_j \cdot x_i &= \min \left\{ p \cdot y \,\middle|\, \sum_{j \in G} y_j = 1, u_i
                \cdot y \geq u_i \cdot x_i \right\}.
            \end{align*}
    \end{itemize}
\end{definition}

\begin{theorem}[\cite{HZ79}]\label{thm:hz}
    An HZ equilibrium always exists.
    Moreover, if $(x, p)$ is an HZ equilibrium, then $x$ is envy-free and Pareto-optimal.
\end{theorem}

Throughout this paper we are always referring to \emph{ex-ante} envy-freeness and Pareto-optimality.
The formal definitions are given below.

\begin{definition}
    Given some allocation $x$, agent $i$ has envy towards agent $i'$ if $u_i \cdot x_i < u_i \cdot
    x_{i'}$.
    $x$ is \emph{envy-free} if no agent has envy towards any other agent.
\end{definition}

\begin{definition}
    Given two allocations $x$ and $y$, $y$ is Pareto-better than $x$ if $u_i \cdot y_i \geq u_i
    \cdot x_i$ for all agents $i$ and $u_i \cdot y_i > u_i \cdot x_i$ for some agent $i$.
    $x$ is \emph{Pareto-optimal} if no other allocation is Pareto-better than it.
\end{definition}

\subsection{Rationality}\label{sec:rationality}

As mentioned earlier, there are instances on which the unique Hylland-Zeckhauser equilibrium
requires an irrational allocation and prices (\cite{VY21}).
In contrast, we will show in this section that there are always EF+PO allocations which are
rational.
This follows from the PPAD membership proof given by \cite{CHR23}, however our
argument is simpler and introduces some basic facts that will be useful in later sections.
The core observation is that fractional perfect matchings which are envy-free and Pareto-optimal can
be characterized polyhedrally.

Let us start by considering the polytope $P_\mathrm{PM}$ of all fractional perfect matchings in the
given market.
\begin{equation*}
    P_\mathrm{PM} \coloneqq \left\{ (x_{i j})_{i \in A, j \in G} \quad \middle| \quad
    \begin{matrix}
        \sum_{j \in G} x_{i j} &= 1 & \forall i \in A, \\
        \sum_{i \in A} x_{i j} &= 1 & \forall j \in G, \\
        x_{i j} &\geq 0 & \forall i \in A, j \in G.
    \end{matrix}
    \right\}
\end{equation*}

It is well-known that Pareto-optimality can be characterized in terms of maximizing along a vector
with strictly positive entries (\cite{Z63}).
Since agents' utilities are linear and the feasible region is a polytope, one can obtain the
corresponding vector in polynomial time using linear programming.

\begin{lemma}\label{lem:pareto_characterization}
    $x^\star \in P_\mathrm{PM}$ is Pareto-optimal if and only if there exist positive $(\alpha_i)_{i
    \in A}$ such that $x^\star$ maximizes $\phi(x) \coloneqq \sum_{i \in A} \alpha_i u_i \cdot x_i$ over
    all $x \in P_\mathrm{PM}$.
    Moreover, if $x^\star$ is rational, $\alpha$ can be computed in polynomial time.
\end{lemma}

See Appendix~\ref{sec:pareto_characterization} for the proof.
Lemma~\ref{lem:pareto_characterization} characterizes the Pareto-optimal allocations.
Moreover, the envy-free allocations themselves form the polytope $P_\mathrm{EF}$ shown below.
\begin{equation*}
    P_\mathrm{EF} \coloneqq \left\{ (x_{i j})_{i \in A, j \in G} \quad \middle| \quad
    \begin{matrix}
        \sum_{j \in G} x_{i j} &= 1 & \forall i \in A, \\
        \sum_{i \in A} x_{i j} &= 1 & \forall j \in G, \\
        u_i \cdot x_i - u_i \cdot x_{i'} &\geq 0 & \forall i, i' \in A, \\
        x_{i j} &\geq 0 & \forall i \in A, j \in G.
    \end{matrix}
    \right\}
\end{equation*}

\begin{theorem}\label{thm:efpo_rational}
    There is always an EF+PO allocation which is a vertex of $P_\mathrm{EF}$ and is thus rational.
\end{theorem}

\begin{proof}
    We know that at least one EF+PO allocation $x^\star$ exists since the HZ equilibrium allocation
    has these properties.
    By Lemma~\ref{lem:pareto_characterization}, $x^\star$ maximizes $\phi(x) \coloneqq \sum_{i \in
    A} \alpha_i u_i \cdot x_i$ over $P_\mathrm{PM}$ for some strictly positive $\alpha$ vector.

    Now consider the linear program $\max \{ \phi(x) \mid x \in P_\mathrm{EF} \}$.
    $P_\mathrm{EF}$ is a polytope so let $x$ be a vertex solution to this LP.
    Clearly $x$ is envy-free since $x \in P_\mathrm{EF}$.
    But since $x^\star$ is also in $P_\mathrm{EF}$ and $P_\mathrm{EF} \subseteq P_\mathrm{PM}$, we
    know that $\phi(x) = \phi(x^\star)$.
    Therefore, by the other direction of Lemma~\ref{lem:pareto_characterization}, $x$ is
    Pareto-optimal.
\end{proof}

\subsection{PPAD-Hardness of Computing EF+PO}\label{sec:ppad_hardness}

We now turn to our main result:

\begin{theorem}\label{thm:efpo_ppad_hard}
    The problem of finding an EF+PO allocation in a one-sided matching market with linear utilities
    is PPAD-hard.
\end{theorem}

Our proof will reduce the problem of finding an approximate HZ equilibrium to that of finding an
EF+PO allocation.
The former was shown to be PPAD hard recently.

\begin{theorem}[\cite{CCPY22}]\label{thm:hz_ppad_hard}
    For any $c > 0$, the problem of finding an $\epsilon$-approximate HZ equilibrium is PPAD-hard
    for $\epsilon \leq 1 / n^c$.
\end{theorem}

There are various reasonable notions of $\epsilon$-approximate HZ equilibria which are polynomially
equivalent.
We will use the following definition.
\footnote{\cite{VY21} give a proof that this notion is indeed equivalent to the one used by
\cite{CCPY22}.}

\begin{definition}
    An assignment $(x_{i j})_{i \in A, j \in G}$ together with non-negative prices $(p_j)_{j \in G}$
    are an \emph{$\epsilon$-approximate HZ equilibrium} if and only if
    \begin{itemize}
        \item each agent $i$ satisfies $\sum_{j \in G} x_{i j} \in [1 - \epsilon, 1]$,
        \item each good $j$ satisfies $\sum_{i \in A} x_{i j} \in [1 - \epsilon, 1]$,
        \item each agent $i$ spends at most $1$, i.e.\ $p \cdot x_i \leq 1$,
        \item each agent $i$ gets a bundle which is at most $\epsilon$ worse than an optimal
            bundle, i.e.\
            \[
                u_i \cdot x_i \geq \max \left \{  u_i \cdot y \;\middle|\; \sum_{j \in G} y_j = 1, p
                \cdot y \leq 1 \right\} - \epsilon.
            \]
    \end{itemize}
\end{definition}

Note that Chen et al.\ assume that all utilities lie in $[0, 1]$ and we will do the same for now.
Additionally, we remark that there is no requirement that the bundle of an agent be (approximately)
cheapest.
This condition is necessary to guarantee some form of approximate Pareto-optimality.
However, it is not needed for the hardness proof and removing it only makes the theorem stronger.

\subsubsection*{Overview of the Reduction}

The general strategy of the reduction consists of the following five steps.

\begin{enumerate}[label=\textbf{Step \arabic*:},left=2em]
    \item We will modify the instance to make sure that all EF+PO allocations are approximate HZ
        equilibria while making sure that HZ equilibria are not perturbed too much.
    \item Starting with an EF+PO allocation $x$ in the modified instance, we will find prices $p$
        and budgets $b$ that make $x$ into a competitive equilibrium using a version of the Second
        Welfare Theorem.
    \item We will use the envy-freeness of $x$ to prove that agents with almost-equal utilities
        have almost-equal budgets in a quantifiable sense.
    \item Then, we will exploit the structure of our modified instance and the linearity of the
        agents' utilities to prove that \emph{all} budgets are almost equal which makes $(x,
        p)$ an approximate HZ-equilibrium.
    \item Finally, we will transform $(x, p)$ to an approximate HZ equilibrium $(\hat{x}, \hat{p})$
        in the original instance, finishing the reduction.
\end{enumerate}

Steps 1, 2, and 5 can be carried out in polynomial time as is required in order to get a polynomial
reduction from approximate HZ to EF+PO. Steps 3 and 4 are the crux of the correctness proof. If two
agents have \emph{equal} utilities and are \emph{non-satiated}, i.e.\ they are not getting 1 unit of
their maximum utility goods, it is not hard to see that their budgets must be equal. Otherwise, the
agent with the smaller budget would necessarily envy the budget with the larger budget. Of course,
these conditions are very strong and not typically satisfied between two arbitrary agents in an
arbitrary instance which is why a modified instance and additional ideas are needed.

\subsubsection*{Step 1: Construction of the Modified Instance}

Our modified instance is going to ensure that between any two agents $i$ and $i'$, there is
a sequence of agents $i = i^{(0)}, \ldots, i^{(l)} = i'$ such that the utilities of $i^{(t)}$ and
$i^{(t + 1)}$ are \emph{almost} the same for all $t$. If we can show that $i^{(t)}$ and $i^{(t +
1)}$ must have almost the same budget for all $t$, then perhaps we can show that $i$ and $i'$ must
have almost the same budget. Moreover, we will ensure that no agent can be satiated. In order to
carry out this construction without perturbing approximate HZ equilibria too much, we will need to
create many copies of identical agents and identical goods.

\begin{definition}
    If two agents have identical utilities for all goods, we say that they are of the \emph{same
    type}.
    Likewise, two goods are of the \emph{same type} if all agents have identical utilities for them.
\end{definition}

Fix a positive integer $k \in \bbN^+$ and some $\epsilon > 0$ such that $k$ is divisible by $n$ and
$k \geq \frac{n^3}{\epsilon}$.
Then we will create a new instance $I' = (A', G', u')$ as follows.
\begin{enumerate}
    \item For each good in $G$, we add $k$ identical copies of said good to $G'$. Likewise, for each
        agent in $A$, we add $k$ identical copies of said agent to $A'$. These copies will allow us
        to add small amounts of new agents and goods without perturbing the HZ equilibria in the
        instance.
    \item Add $k / n$ identical goods for which every agent has utility 2 which is double of what
        they can get from goods in $G$.
        For this reason, we call these \emph{awesome goods} and their limited quantity is going to
        prevent satiation.
    \item For each pair $\{i, i'\}$ of distinct agents in $A$, we create a sequence of
        \emph{interpolating agents}.
        Order the \emph{types} of goods in some arbitrary way $\{t_1, \ldots,
        t_{n}\}$.
        Now add up to $\frac{1}{\epsilon}$ agents by starting with with the utilities of agent $i$
        and slowly increasing / decreasing the utility for $t_1$ goods in steps of $\epsilon$ until
        we reach the utility that $i'$ has for $t_1$ goods.
        Repeat this process for $t_2, \ldots, t_n$.
        The final result of this procedure will be at most $\frac{n}{\epsilon}$ additional
        agents which slowly interpolate between the utilities of $i$ and $i'$, one type of good at a
        time.
        See Figure~\ref{fig:interp}.
    \item Finally, add dummy agents to $A'$ until $|A'| = |G'|$.
        These agents have identical utilities for all goods.
        Note that we added fewer interpolating agents than awesome goods since $k >
        \frac{n^3}{\epsilon}$.
\end{enumerate}

\begin{figure}[htb]
    \centering
    \begin{tikzpicture}[xscale=1.5, yscale=1.5]
        \draw[->, thick] (0, 0) -- (6, 0);
        \draw[->, thick] (0, 0) -- (0, 4);

        \node[fill, circle, inner sep=2pt, label=right:$i$, red] (v) at (1, 3) {};
        \node[fill, circle, inner sep=2pt, label=right:$i'$, red] (w) at (5, 1) {};

        \node[right] at (6, 0) {$j$};
        \node[right] at (0, 4) {$j'$};

        \draw[dotted] (v) -- (1,0) node[at start, right] {} node[at end, below] {$u_{i j}$};
        \draw[dotted] (v) -- (0,3) node[at start, right] {} node[at end, left] {$u_{i j'}$};

        \draw[dotted] (w) -- (5,0) node[at start, right] {} node[at end, below] {$u_{i' j}$};
        \draw[dotted] (w) -- (0,1) node[at start, right] {} node[at end, left] {$u_{i' j'}$};

        \draw[decorate,decoration={brace,amplitude=5pt},xshift=-2pt] (1, 1.8) -- (1, 2.2)
        node[left,midway,xshift=-4pt] {$\epsilon$};

        \draw[decorate,decoration={brace,amplitude=5pt,mirror},yshift=-2pt] (2.6, 1) -- (3, 1)
        node[below,midway,yshift=-4pt] {$\epsilon$};

            \foreach \t in {0.2,0.4,0.6,0.8,1}{
                \pgfmathsetmacro{\x}{1} 
                \pgfmathsetmacro{\y}{3+\t*(1-3)} 
                \node[fill, circle, inner sep=1pt] at (\x,\y) {};
            }
            
            \foreach \t in {0.1,0.2,0.3,0.4,0.5,0.6,0.7,0.8,0.9,1}{
                \pgfmathsetmacro{\x}{5+\t*(1-5)} 
                \pgfmathsetmacro{\y}{1} 
                \node[fill, circle, inner sep=1pt] at (\x,\y) {};
            }
    \end{tikzpicture}
    \caption{For each pair of agents $i$ and $i'$ (large red dots) we add \emph{interpolating
    agents} (small black dots) to transition between the utility vector $u_i$ and $u_{i'}$ in small
    steps. This is done coordinate-wise and this figure depicts an example with only two goods $j$
    and $j'$.\label{fig:interp}}
\end{figure}
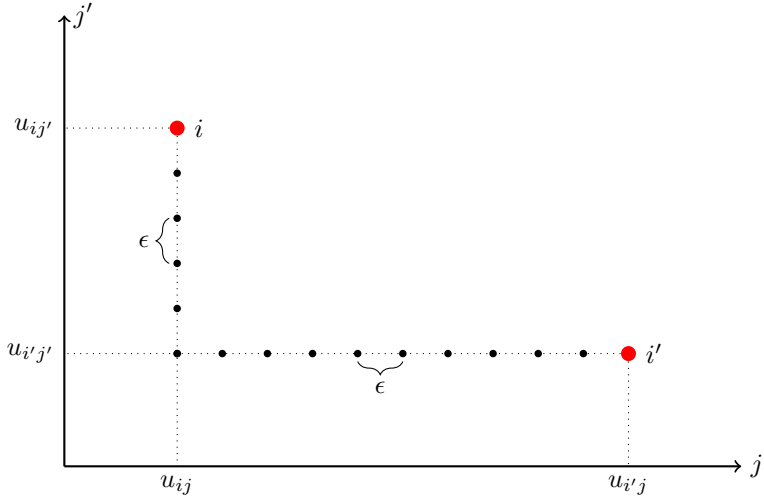

\begin{lemma}\label{lem:n_prime_bound}
    Let $n' = |A'| = |G'|$ be the number of agents / goods in the modified instance.
    Then $n' \leq 2 k n$.
\end{lemma}

\begin{proof}
    We add $k n$ goods through identical copies of the agents in $A$ and $k / n \leq k n$ awesome
    goods.
\end{proof}

\subsubsection*{Step 2: Finding Prices and Budgets}

In the following assume that we are given a rational EF+PO allocation $x$ on $I'$ which is encoded
with a polynomial number of bits.
Our goal will be to construct an approximate HZ solution on $I$.
We now carry out Step~2 by finding budgets and prices that make $x$ a competitive equilibrium on
$I'$.
Recall that by Lemma~\ref{lem:pareto_characterization}, there exist positive $\alpha_i$ for all $i
\in A'$ such that $x$ solves
\begin{maxi*}
    {}
    {\sum_{i \in A'} \alpha_i u_i \cdot x_i}
    {}
    {}
    \addConstraint{\sum_{i \in G'} x_{i j}}{= 1}{\quad \forall i \in A'}
    \addConstraint{\sum_{j \in A'} x_{i j}}{= 1}{\quad \forall j \in G'}
    \addConstraint{x_{i j}}{\geq 0}{\quad \forall i \in A', j \in G'.}
\end{maxi*}

Moreover, we can find such $\alpha_i$ in polynomial time since we obtained them using an LP in the
proof of Lemma~\ref{lem:pareto_characterization}.
Consider now an optimal solution $(p, q)$ to the dual.
\begin{mini*}
    {}
    {\sum_{i \in A'} q_i + \sum_{j \in G'} p_j}
    {}
    {}
    \addConstraint{q_i + p_j}{\geq \alpha_i u_{i j}}{\quad \forall i \in A', j \in G'}
\end{mini*}
and define $b_i \coloneqq \alpha_i u_i \cdot x_i - q_i$ to be budget of agent $i$.
Note that we may assume that $p, q \geq 0$ since all utilities are non-negative.
As shown in Lemma~\ref{lem:po_optimal_bundles}, $x$ really is a competitive equilibrium with prices
$p$ and budgets $b$.

\begin{lemma}\label{lem:po_optimal_bundles}
    For every agent $i$, we have that $b_i \geq 0$ and $x_i$ is an optimum solution to
    \begin{maxi*}
        {}
        {u_i \cdot x_i}
        {}
        {}
        \addConstraint{\sum_{j \in G'} x_{i j}}{\leq 1}
        \addConstraint{p \cdot x_i}{\leq b_i}
        \addConstraint{x_i}{\geq 0.}
    \end{maxi*}
\end{lemma}

\begin{proof}
    First, observe that
    \[
        \sum_{j \in G'} p_j x_{i j} = \sum_{j \in G'} (\alpha_i u_{i j} - q_i) x_{i j} = \alpha_i
        u_i \cdot x_i - q_i = b_i
    \]
    using complimentary slackness and the fact that $\sum_{j \in G'} x_{i j} = 1$.
    So $x_i$ is at least feasible and clearly $b_i \geq 0$ since prices are non-negative.

    Now take any feasible solution $(y_j)_{j \in G'}$ of the LP.
    Then
    \[
        \sum_{j \in G'} u_{i j} y_j \leq \sum_{j \in G'} \frac{p_j + q_i}{\alpha_i} y_j
        \leq \frac{b_i + q_i}{\alpha_i}
        = u_i \cdot x_i
    \]
    by dual feasibility and the definition of $b_i$.
\end{proof}

\subsubsection*{Step 3: Almost Equality of Budgets via Envy-Freeness}

Our goal will now be to use envy-freeness in order to show that agents' budgets are approximately
equal.
First, we need to prove several simple lemmas which ultimately allow us to prove that no agent is
satiated (in a quantifiable way).

\begin{lemma}\label{lem:equal_goods}
    If $j$ and $j'$ are goods of the same type, then $p_j = p_{j'}$.
\end{lemma}

\begin{proof}
    Note that every good is fully matched.
    But if the prices were different, then any agent matched to the more expensive good would be
    violating Lemma~\ref{lem:po_optimal_bundles}.
\end{proof}

\begin{lemma}\label{lem:not_satiated}
    For any non-dummy agent $i$, $u_i \cdot x_i \leq 1.6$.
    In particular $i$ is not satiated.
\end{lemma}

\begin{proof}
    If this were not the case, $i$ would need to get at least $0.6$ units of an awesome good.
    But then any other non-dummy agent must get at least $0.1$ units of an awesome good by
    envy-freeness.
    Since there are much more non-dummy agents than awesome goods, this is a contradiction.
\end{proof}

\begin{lemma}\label{lem:budget_positive}
    There exists at least one non-dummy agent $i$ with $b_i > 0$.
\end{lemma}

\begin{proof}
    There must be at least one non-dummy agent $i$ who buys a positive fraction of an awesome good.
    This is because if a non-dummy agent received any amount of an awesome good, this would violate
    Pareto-optimality since they could swap goods with a non-dummy agent.
    But since $i$ is not satiated by Lemma~\ref{lem:not_satiated}, the price of said awesome
    good must be positive and so must the agent's budget.
\end{proof}

In particular, we can rescale all $\alpha$, $p$, $q$, and $b$ so that the maximum budget of any
non-dummy agent is exactly 1.
In the remainder of this section, we assume that this is the case.

\begin{lemma}
    If $i$ and $i'$ are agents of the same type, then $b_i = b_{i'}$.
\end{lemma}

\begin{proof}
    Note that by Lemma~\ref{lem:not_satiated}, no agent receives their maximum possible utility.
    So if $b_i \neq b_{i'}$, assume wlog.\ that $b_i < b_{i'}$.
    Then since $i'$ is optimally spending $b_{i'}$ and both agents agree on the utilities of all
    goods, both agents agree that $i'$ is getting a higher utility bundle than $i$.
    Thus $i$ would be envious.
\end{proof}

Now that we have established several basic facts about the budgets and bundles of the agents, we
will turn to our main objective: show that the budgets are almost equal.
As mentioned in our high level plan, we will first show that two agents whose utility vectors are
almost equal, must have almost equal budgets.
This is done in Lemmas \ref{lem:alpha_i_bounded} and \ref{lem:general_diff} below.

\begin{lemma}\label{lem:alpha_i_bounded}
    For any non-dummy agent $i$ we have $\alpha_i \leq 5 n^2$.
\end{lemma}

\begin{proof}
    Consider an awesome good $j^\star$.
    By dual feasibility, we know that $p_{j^\star} + q_i \geq \alpha_i u_{i j^\star} = 2 \alpha_i$.
    But on the other hand, note that
    \[
        q_i = \alpha_i u_i \cdot x_i - b_i \leq \alpha_i u_i \cdot x_i \leq 1.6 \alpha_i
    \]
    using Lemma~\ref{lem:po_optimal_bundles} and Lemma~\ref{lem:not_satiated}.
    Comining these inequalities we get $p_{j^\star} \geq 0.4 \alpha_i$.

    Lastly, we note that the $k / n$ awesome goods can only be sold to the non-dummy agents of which
    there are at most $2 k n$ and each of which has a budget of at most 1 after rescaling. So the
    price of the awesome goods must be at most $2 n^2$ which finishes the proof.
\end{proof}

\begin{lemma}\label{lem:general_diff}
    Let $i, i'$ be two non-dummy agents whose utilities are identical except for the goods of one
    type where they differ by at most $\epsilon$.
    Then $|b_i - b_{i'}| \leq 5 n^2 \epsilon$.
\end{lemma}

\begin{proof}
    Note that since the $u_i$ and $u_{i'}$ disagree only by epsilon, we have
    \[
        u_i \cdot x_{i'} \geq u_{i'} \cdot x_{i'} - \epsilon \geq u_{i'} \cdot x_i - \epsilon \geq
        u_i \cdot x_i - 2 \epsilon
    \]
    using envy-freeness.
    In fact, depending on whether $u_{i}$ or $u_{i'}$ has the higher utility, we can only lose an
    $\epsilon$ in the first or the last inequality.
    So we actually get $u_i \cdot x_{i'} \geq u_i \cdot x_i - \epsilon$.

    Now we can compute
    \begin{align*}
        b_{i'} &= \sum_{j \in G} x_{i' j} p_j \\
        &=\sum_{j \in G} x_{i' j} (\alpha_i u_{i j} - q_i) \\
        &= \alpha_i u_i \cdot x'_i - q_i \\
        &\geq \alpha_i u_i \cdot x' - \epsilon \alpha_i - q_i \\
        &= b_i - \epsilon \alpha_i
    \end{align*}
    and using symmetry and Lemma~\ref{lem:alpha_i_bounded} we conclude
    $ |b_i - b_{i'}| \leq \epsilon \max \{\alpha_i, \alpha_{i'}\} \leq 5n^2 \epsilon. $
\end{proof}

Lemma~\ref{lem:general_diff} is enough to show that the difference in budgets between ``close''
agents tends to zero for an inverse-polynomial $\epsilon$.
However, between any two distinct agents $i, i' \in A$, it takes us up to $\frac{n^2}{\epsilon}$
agents to interpolate between them and therefore we cannot give any non-trivial bound on the
difference in budget between arbitrary agents.
It seems as if we have not won anything!

\subsubsection*{Step 4: Bounding the Budget Changes for Interpolating Agents}

The key argument that makes our construction work is as follows: we are going to show that along any
chain of interpolating agents, the budgets cannot change more than $O(n^2)$ many times due to
the linearity of the utilities.
Before we prove this in full generality, it is insightful to consider a simpler situation in which
agents do not have the matching constraint.
Without the matching constraint, the optimal thing to do for any agent is to spend their entire
budget on whichever goods have the maximum ``bang per buck'', i.e.\ those goods $j$ that
maximize $\frac{u_{i j}}{p_j}$.

It is not hard to see that when two agents agree on \emph{which} goods are maximum bang per buck,
then their budgets must be equal.
Otherwise, the agent with the larger budget would be able to buy more of those goods and thus would
be envied by the agent with the smaller budget.
When we modify the utility of one good, the set of maximum bang per buck goods can only change
twice.
See Figure~\ref{fig:bang_per_buck}.

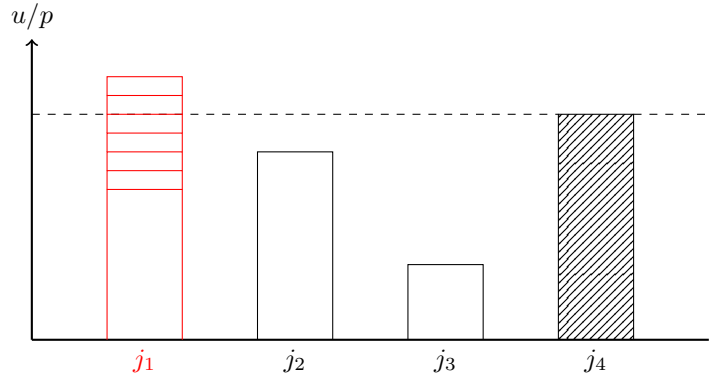
\begin{figure}[htb]
    \centering
    \begin{tikzpicture}
        \draw[thick, -] (0, 0) -- (9, 0);
        \draw[thick, ->] (0, 0) -- (0, 4);

        \node[above] at (0, 4) {$u / p$};

        \draw[-, dashed] (0, 3) -- (7, 3);
        \draw[-, dashed] (8, 3) -- (9, 3);
        \draw[-, red] (1, 0) -- (1, 2) -- (2, 2) -- (2, 0);
        \draw[-, red] (1, 2) -- (1, 2.25) -- (2, 2.25) -- (2, 2);
        \draw[-, red] (1, 2.25) -- (1, 2.5) -- (2, 2.5) -- (2, 2.25);
        \draw[-, red] (1, 2.5) -- (1, 2.75) -- (2, 2.75) -- (2, 2.5);
        \draw[-, red] (1, 2.75) -- (1, 3) -- (2, 3) -- (2, 2.75);
        \draw[-, red] (1, 3) -- (1, 3.25) -- (2, 3.25) -- (2, 3);
        \draw[-, red] (1, 3.25) -- (1, 3.5) -- (2, 3.5) -- (2, 3.25);

        \fill[pattern=north east lines] (7, 0) -- (7, 3) -- (8, 3) -- (8, 0);

        \draw[-] (3, 0) -- (3, 2.5) -- (4, 2.5) -- (4, 0);
        \draw[-] (5, 0) -- (5, 1) -- (6, 1) -- (6, 0);
        \draw[-] (7, 0) -- (7, 3) -- (8, 3) -- (8, 0);

        \node[below, red] at (1.5, 0) {$j_1$};
        \node[below] at (3.5, 0) {$j_2$};
        \node[below] at (5.5, 0) {$j_3$};
        \node[below] at (7.5, 0) {$j_4$};
    \end{tikzpicture}
    \caption{Shown is an agent who is interested in goods $j_1$ to $j_4$ which are plotted by
    their bang per buck. If we change only the utility of good $j_1$ (red) and leave the rest the
    same, there are only three possible sets of maximum bang per buck goods: $\{j_4\}$, $\{j_1,
    j_4\}$, and $\{j_1\}$. So along any chain of interpolating agents where we change only the
    utility for $j_1$ (monotonically), there will be at most two times that the set of maximum bang
    per buck goods and with it the budget of the agent can change.\label{fig:bang_per_buck}}
\end{figure}

Unfortunately, once we add in the matching constraint which is crucial to our setting, this simple
characterization no longer works.
The core issue is that with the matching constraint, the optimal bundles of an agent depend not just
on the utilities and prices of the goods but also on the budget of the agent.
Since our goal is to show that agents have identical budgets, this easily leads to circular
reasoning.
The way around this is to instead assume that agents have the same optimal bundles for all
\emph{potential} budgets.

\begin{definition}
    For any agent $i$, define a function $\theta_i(t)$ which maps any $t \geq 0$ to the set of all
    goods $j \in G$ such that $y_j$ can be positive in an optimum solution to
    \begin{maxi}
        {}
        {u_i \cdot y}
        {}
        {\label{lp:opt_bundle}}
        \addConstraint{\sum_{j \in G'} y}{\leq 1}
        \addConstraint{p \cdot y}{\leq t}
        \addConstraint{y}{\geq 0.}
    \end{maxi}

    In other words, $\theta_i(t)$ are the goods which can participate in an optimal bundle for agent
    $i$ at budget $t$.
\end{definition}

\begin{lemma}\label{lem:special_diff}
    Let $i, i'$ be two agents with $\theta_i = \theta_{i'}$, then $b_i = b_{i'}$.
\end{lemma}

\begin{proof}
    Assume otherwise and let $b_i < b_{i'}$ wlog.
    We will show that $i$ must envy $i'$.

    Consider LP \eqref{lp:opt_bundle} with $t = b_{i'}$ which maximizes the utility of agent $i$ but
    under the higher budget of agent $i'$.
    We claim that $y = x_{i'}$ is an optimum solution of this LP.
    To see this, consider the dual as well:
    \begin{mini}
        {}
        {\mu + \rho b_i}
        {}
        {\label{lp:bundle_dual}}
        \addConstraint{\mu + p_j \rho}{\geq u_{i j}}
        \addConstraint{\mu, \rho}{\geq 0.}
    \end{mini}

    Now, for any $j$, we know that if $x_{i' j} > 0$, then $j \in \theta_{i'}$ by definition.
    But since $\theta_{i'} = \theta_i$, this implies that there is some optimal primal solution
    with $y_j > 0$.
    By complementary slackness, this implies that $\mu + p_j \rho = u_{i j}$.
    Therefore, $x_{i'}$ is a feasible solution to the LP which, together with $\mu$ and $\rho$, satisfies
    the complementary slackness conditions and is therefore optimal.

    Finally, since no agent is satiated (Lemma~\ref{lem:not_satiated}), increasing the budget always
    increases the optimum value of the LP, implying that $u_i \cdot x_i < u_i \cdot x_{i'}$.
    This contradicts envy-freeness.
\end{proof}

\begin{lemma}\label{lem:limited_theta}
    Let $i_1, \ldots, i_m$  be a set of agents such that all agents agree on all utilities except
    for possibly one type of good.
    Then $|\{\theta_{i_1}, \ldots, \theta_{i_m}\}| \leq 2n + 1$.
\end{lemma}

\begin{proof}
    We will give a geometric proof of this fact.
    First, we will need to understand the behavior of any particular $\theta_i(t)$.
    We are interested in the goods which can be used in an optimum solution $y$ to
    \eqref{lp:opt_bundle}.
    By complementary slackness these are the goods for which the corresponding dual constraint is
    tight in the dual \eqref{lp:bundle_dual}.

    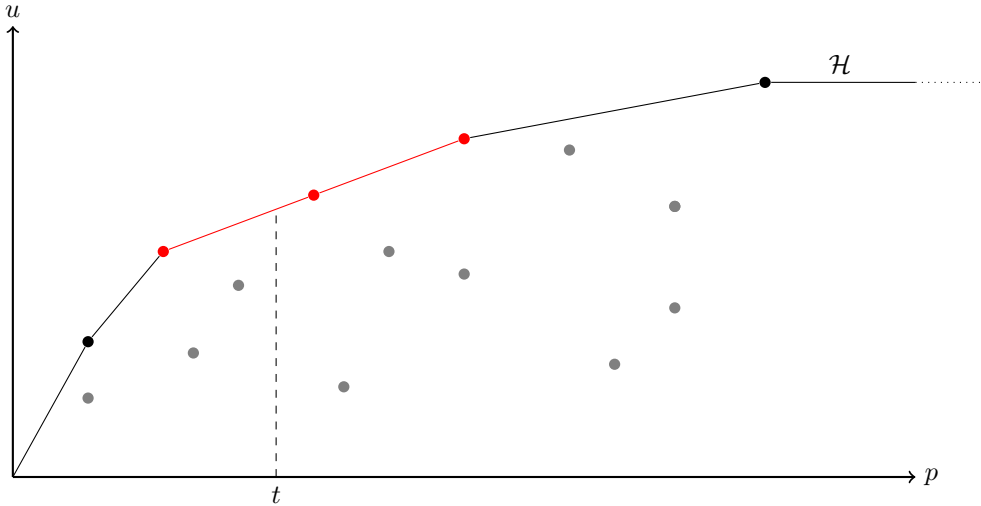
\begin{figure}[hbt]
        \centering
        \begin{tikzpicture}[xscale=2, yscale=1.5]


            \node[circle, fill, inner sep=1.5pt] (h1) at (0.5, 1.2) {};
            \node[circle, fill, inner sep=1.5pt, red] (h2) at (1, 2) {};
            \node[circle, fill, inner sep=1.5pt, red] (h3) at (2, 2.5) {};
            \node[circle, fill, inner sep=1.5pt, red] (h4) at (3, 3) {};
            \node[circle, fill, inner sep=1.5pt] (h5) at (5, 3.5) {};

            \node[circle, fill, inner sep=1.5pt, gray] (i1) at (0.5, 0.7) {};
            \node[circle, fill, inner sep=1.5pt, gray] (i2) at (1.2, 1.1) {};
            \node[circle, fill, inner sep=1.5pt, gray] (i3) at (1.5, 1.7) {};
            \node[circle, fill, inner sep=1.5pt, gray] (i4) at (2.2, 0.8) {};
            \node[circle, fill, inner sep=1.5pt, gray] (i5) at (2.5, 2) {};
            \node[circle, fill, inner sep=1.5pt, gray] (i6) at (3, 1.8) {};
            \node[circle, fill, inner sep=1.5pt, gray] (i7) at (3.7, 2.9) {};
            \node[circle, fill, inner sep=1.5pt, gray] (i8) at (4, 1) {};
            \node[circle, fill, inner sep=1.5pt, gray] (i9) at (4.4, 1.5) {};
            \node[circle, fill, inner sep=1.5pt, gray] (i10) at (4.4, 2.4) {};
            \node[circle, fill, inner sep=1.5pt, gray] (i11) at (4.4, 2.4) {};

            \draw[-] (0, 0) -- (h1) -- (h2);
            \draw[-, red] (h2) -- (h3) -- (h4);
            \draw[-] (h4) -- (h5) -- (6, 3.5);
            \draw[-, dotted] (6, 3.5) -- (6.5, 3.5);

            \draw[thick, ->] (0, 0) -- (6, 0);
            \draw[thick, ->] (0, 0) -- (0, 4);

            \node[right] at (6, 0) {$p$};
            \node[above] at (0, 4) {$u$};

            \draw[-, dashed] (1.75, 0) -- (1.75, 2.375);
            \node[below] at (1.75, 0) {$t$};

            \node[above] at (5.5, 3.5) {$\calH$};

        \end{tikzpicture}
        \caption{Depicted is $\calH$ and its relationship to optimal bundles. Each point represents a
        good or collection of goods with identical price and utility. Gray points are dominated and
        will never be part of an optimal bundle. Points on $\calH$ can be part of an optimal bundle
        depending on the budget $t$. A typical case is shown in which $\theta_i(t)$
        consists of the three red goods that lie on the edge of $\calH$ which corresponds to the
        tight dual constraints at budget $t$.\label{fig:dual_characterization}}
    \end{figure}

    Now let us interpret this dual geometrically in $\bbR^2$.
    The expression $\mu + \rho t$ represents a line in $t$ with non-negative slope.
    The condition that $\mu + p_j \rho \geq u_{i j}$ means that this line lies above the point
    $(p_j, u_{i j})$.
    In other words, the dual objective function for a fixed $t$ is optimized by a line which is as
    low as possible at $t$ and yet lies above all the points $(p_j, u_{i j})$.
    This characterizes precisely the upper boundary of the convex hull of the point set
    \[
        \{(0, 0)\} \cup \{(p_j, u_{i j}) \mid j \in G'\} \cup \{(\infty, \max_{j \in G'} u_{i j})\}
    \]
    which we will denote by $\calH$.

    And together with what we already know from complementary slackness, this gives a nice geometric
    characterization of $\theta_i$.
    For a given $t$, consider the point $(t, v) \in \calH$.
    If $(t, v)$ is a vertex of the convex hull, i.e.\ corresponds to $(p_j, u_{i j})$ for some good
    $j \in G'$, then only this good---or more precisely only goods with identical price and
    utility---can participate in an optimum bundle.
    On the other hand, if $(t, v)$ is not a vertex, then it lies on some line $L$ that bounds the
    convex hull (determined by at least two linearly independent tight dual constraints).
    $\theta_i(t)$ will then consist of all those goods j such that $(p_j, u_{i j})$ lies on $L$.
    See Figure~\ref{fig:dual_characterization}.

    Let us now return to the agents $i_1, \ldots, i_m$ and consider what happens to $\calH$ when we
    shift a single point along the $y$-axis.
    By the characterization of $\theta$, the only thing we need to know to uniquely determine
    $\theta$ is which goods lie on $\calH$ and out of these which goods are vertices of $\calH$.
    Call this data the \emph{structure} of $\calH$.

    Let $j$ be the type of good for which the agents have differing utilities.
    When we remove $j$, we can construct a convex hull $\calH'$ on the rest of the goods
    (corresponding to optimal bundles without type $j$).
    Finally, observe that the structure of $\calH$ only depends on the relationship (below,
    intersecting, above) which $(p_j, u_{i j})$ has with the at most $n$ lines that bound $\calH'$.
    Since there are only $2 n + 1$ possible ways in which a point can relate to $n$ lines, this
    proves the claim.
    See Figure~\ref{fig:convex_hull_change}.
\end{proof}

\begin{figure}[hbt]
    \centering
    \begin{tikzpicture}[xscale=2, yscale=1.5]

        \node[circle, fill, inner sep=1.5pt] (h1) at (0.5, 1.2) {};
        \node[circle, fill, inner sep=1.5pt] (h2) at (1, 2) {};
        \node[circle, fill, inner sep=1.5pt] (h4) at (3, 3) {};
        \node[circle, fill, inner sep=1.5pt] (h5) at (5, 3.5) {};

        \node[circle, fill, inner sep=1.5pt, gray] (i1) at (0.5, 0.7) {};
        \node[circle, fill, inner sep=1.5pt, gray] (i2) at (1.2, 1.1) {};
        \node[circle, fill, inner sep=1.5pt, gray] (i3) at (1.5, 1.7) {};
        \node[circle, fill, inner sep=1.5pt, gray] (i4) at (2.2, 0.8) {};
        \node[circle, fill, inner sep=1.5pt, gray] (i5) at (2.5, 2) {};
        \node[circle, fill, inner sep=1.5pt, gray] (i6) at (3, 1.8) {};
        \node[circle, fill, inner sep=1.5pt, gray] (i7) at (3.7, 2.9) {};
        \node[circle, fill, inner sep=1.5pt, gray] (i8) at (4, 1) {};
        \node[circle, fill, inner sep=1.5pt, gray] (i9) at (4.4, 1.5) {};
        \node[circle, fill, inner sep=1.5pt, gray] (i10) at (4.4, 2.4) {};
        \node[circle, fill, inner sep=1.5pt, gray] (i11) at (4.4, 2.4) {};

        \draw[-] (0, 0) -- (h1) -- (h2);
        \draw[-] (h2) -- (h4);
        \draw[-] (h4) -- (h5) -- (6, 3.5);
        \draw[-, dotted] (6, 3.5) -- (6.5, 3.5);

        \draw[thick, ->] (0, 0) -- (6, 0);
        \draw[thick, ->] (0, 0) -- (0, 4);

        \node[right] at (6, 0) {$p$};
        \node[above] at (0, 4) {$u$};


        \node[above] at (5.5, 3.5) {$\calH'$};

        \node[circle, fill, inner sep=1.5pt, red] (h3_1) at (2, 2) {};
        \node[circle, fill, inner sep=1.5pt, red] (h3_2) at (2, 2.6) {};
        \node[circle, fill, inner sep=1.5pt, red] (h3_3) at (2, 3.1) {};
        \node[circle, fill, inner sep=1.5pt, red] (h3_4) at (2, 4) {};

        \draw[-, red] (h2) -- (h3_2) -- (h4);

        \draw[-, red] (h2) -- (h3_3) -- (h5);
        \draw[-, red] (h1) -- (h3_4) -- (6, 4);
        \draw[-, red, dotted] (6, 4) -- (6.4, 4);
        \node[above, red] at (5.5, 4) {$\calH$};

        \draw[-, dotted] (0.5, 1.2) -- (2, 4.8);
        \draw[-, dotted] (1, 2) -- (2.75, 4.8);
        \draw[-, dotted] (3, 3) -- (0, 2.25);
        \draw[-, dotted] (0, 3.5) -- (5, 3.5);

    \end{tikzpicture}
    \caption{Shown are several convex hulls $\calH$ (red) as the red good's utility is changed.
    Note that the structure of $\calH$ only changes when we cross one of the bounding lines of
    $\calH$ -- the convex hull without the red good.\label{fig:convex_hull_change}}
\end{figure}
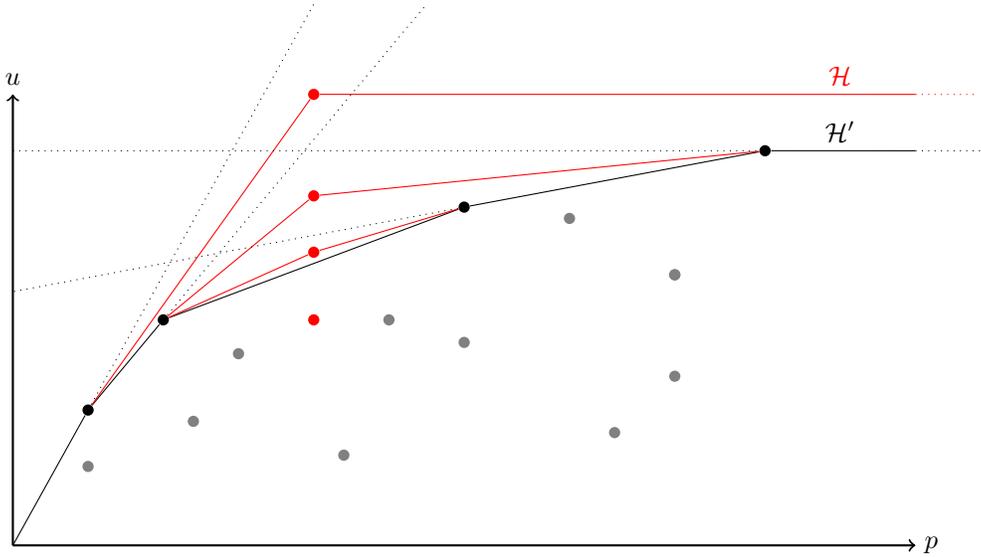

\begin{lemma}\label{lem:diff_bound}
    Let $i, i'$ be two non-dummy agents.
    Then $|b_i - b_{i'}| \leq 5 \epsilon n^4$.
\end{lemma}

\begin{proof}
    Consider the chain of interpolating agents between $i$ and $i'$.
    There can be at most $n$ types of goods on which $i$ and $i'$ have different utilities.
    So we can divide these agents into at most $n$ groups inside of which the agents differ only on
    one good.
    By Lemma~\ref{lem:limited_theta}, inside each group there are at most $2n + 1$ different $\theta$
    functions.
    By Lemma~\ref{lem:special_diff}, the budgets of agents who have identical $\theta$ functions must
    be identical.
    And so there are at most $2n$ opportunities for $\theta$ to change inside each group, totaling
    to $2n^2$ changes overall.
    Each of these changes in $\theta$ corresponds to two agents that differ in their utilities by at
    most $\epsilon$ on one good, thus Lemma~\ref{lem:general_diff} applies and we get $|b_i -
    b_{i'}| \leq 2n^2 \cdot 5 \epsilon n^2$.
\end{proof}

\subsubsection*{Step 5: Contracting to the Original Instance}

To finish the proof, let us construct our approximate HZ equilibrium $(\hat{x}, \hat{p})$ on the
original instance by contracting the allocation along the copies of goods and agents.
For any $i \in A, j \in G$ let $\hat{x}_{i j}$ be the average over all $x_{i' j'}$ where $i'$ are
the $k$ identical copies of $i$ and $j'$ are the $k$ identical copies of $j$ in $I'$.
The parts of $x$ going to the dummy agents, interpolating goods, and awesome goods are simply
dropped.

\begin{theorem}
    If $\epsilon \leq \frac{1}{5 n^5}$, then $(\hat{x}, \hat{p})$ is a $\frac{3}{n}$-approximate HZ
    equilibrium in the original instance $I$.
\end{theorem}

\begin{proof}
    First, observe that as there are $k$ copies of each agent $i$ and only $k / n$ awesome goods, we
    have that $\sum_{j \in G} \hat{x}_{i j} = [1 - \frac{1}{n}, 1]$.
    Likewise, the total number of interpolating and dummy agents is $k / n$ and there are $k$ copies
    of each good $j$ so $\sum_{i \in A} \hat{x}_{i j} = [1 - \frac{1}{n}, 1]$.
    This establishes that $\hat{x}$ is an approximately perfect fractional matching.

    Moreover, it is clear that no agent overspends as no non-dummy agent spends more than 1 in $I'$
    and we have only removed allocations during the contraction.

    Finally, we need to show that no agent is far from their optimum bundle.
    For that, let $y$ be an optimum solution to 
    \begin{maxi*}
        {}
        {u_i \cdot y}
        {}
        {}
        \addConstraint{\sum_{j \in G} y}{= 1}
        \addConstraint{p \cdot y}{\leq 1}
        \addConstraint{y.}{\geq 0}
    \end{maxi*}
    By Lemma~\ref{lem:diff_bound}, we know that $b_i \geq 1 - \frac{1}{n}$.
    And so $u_i \cdot x_i \geq (1 - 1/n) u_i \cdot y$ since we could otherwise scale down $y$ and
    violate Lemma~\ref{lem:po_optimal_bundles}.
    Note that it is important here that $x_i$ was optimal even among bundles that get \emph{at most}
    one unit of good.

    Lastly, we know that $u_i \cdot \hat{x}_i \geq u_i \cdot x_i - \frac{2}{n}$ since the only thing
    that was lost when contracting were up to $\frac{1}{n}$ awesome goods as mentioned above.
    Thus
    \[
        u_i \cdot \hat{x}_i \geq (1 - 1/n) u_i \cdot y - \frac{2}{n} \geq u_i \cdot y - \frac{3}{n}
    \]
    finishing the proof.
\end{proof}

\begin{proof}[Proof of Theorem~\ref{thm:efpo_ppad_hard}]
    If we choose $\epsilon = \frac{1}{5 n^5}$ and $k = 5 n^8$, then the constructed instance has at
    most $10 n^9$ agents by Lemma~\ref{lem:n_prime_bound}.
    Given a rational EF+PO allocation with polynomial encoding length, we can construct $(\hat{x},
    \hat{p})$ as above in polynomial time and get a $\frac{3}{n}$-approximate HZ equilibrium.
    By Theorem~\ref{thm:hz_ppad_hard}, the latter problem is PPAD-hard.
\end{proof}

Lastly, we remark that Theorem~\ref{thm:efpo_ppad_hard} can be slightly strengthed to show hardness
of computing \emph{approximately} envy-free and Pareto-optimal solutions with inverse polynomial
$\epsilon$.
Lemmas \ref{lem:general_diff} and \ref{lem:special_diff} require minor modifications for the proof
to go through.

\subsection{2-EF and 2-IC via Nash Bargaining}\label{sec:one_sided_nash}

Now that we have seen that finding EF+PO allocations is PPAD-hard, this raises the question: what is
the best that we can actually do in polynomial time?
It turns out that Nash bargaining comes to the rescue here.
\cite{N53} studied the problem of two or more agents bargaining over a common outcome, for
example how they should split up certain goods amongst themselves.
He showed that there is a unique point that satisfies certain axioms\footnote{The four axioms are
Pareto-optimality, symmetry, invariance under affine transformations of utilities, and independence
of irrelevant alternatives.} and moreover that this point is characterized as maximizing the product
of the agents' utilities, i.e.\ the Nash social welfare.

In our case, this means that the Nash bargaining solution is given by the solution to
\begin{maxi}
    {}
    {\prod_{i \in A} u_i \cdot x_i}
    {}
    {\label{cp:nash_one_sided}}
    \addConstraint{\sum_{j \in G} x_{i j}}{= 1}{\quad \forall i \in A}
    \addConstraint{\sum_{i \in A} x_{i j}}{= 1}{\quad \forall j \in G}
    \addConstraint{x_{i j}}{\geq 0}{\quad \forall i \in A, j \in G.}
\end{maxi}

Since the objective function is log-concave, general purpose convex programming techniques can be
used to find approximate solutions to this program which is a stark difference to HZ.
For this reason, \cite{HV21} proposed Nash bargaining as an alternate solution
concept for cardinal-utility matching markets of various kinds.
We strengthen the case for Nash bargaining as an HZ alternative by showing that Nash
bargaining points are approximately envy-free and approximately incentive compatible.

\begin{definition}
    An allocation $(x_{ij})_{i \in A, j \in G}$ is \emph{$\alpha$-approximately envy-free} or just
    $\alpha$-EF if for every $i, i' \in A$ we have $u_i \cdot x_i \geq \frac{1}{\alpha} u_i \cdot
    x_{i'}$.
    In other words, no agent envies another agent by more than a factor of $\alpha$.
\end{definition}

\begin{theorem}\label{thm:nash_2_ef}
    Let $x$ be an optimum solution to \eqref{cp:nash_one_sided}.
    Then $x$ is 2-EF.
\end{theorem}

\begin{proof}
    Assume otherwise, i.e.\ that there are agents $i, i' \in A$ such that $u_i \cdot x_{i'} = \alpha
    u_i \cdot x_{i}$ and $\alpha > 2$.
    Then we consider what happens when we swap some $\epsilon$-fraction of the bundle that $i$ gets
    with the bundle that $i'$ gets.
    This maintains feasibility.

    By doing so, the product of the agents' utilities changes by a factor of
    \[
        \frac{(u_i \cdot x_i (1 - \epsilon) + \alpha u_i \cdot x_i \epsilon)(u_{i'} \cdot x_{i'} (1 -
        \epsilon) + u_{i'} \cdot x_i \epsilon)}{(u_i \cdot x_i) (u_{i'} \cdot x_{i'})}.
    \]
    We now evaluate the derivative of this expression wrt.\ to $\epsilon$ at $\epsilon = 0$ and get
    \[
        \frac{(\alpha - 1) (u_i \cdot x_i) (u_{i'} \cdot x_{i'}) + (u_i \cdot x_i)(u_{i'} \cdot x_i -
        u_{i'} \cdot x_{i'})}{(u_i \cdot x_i) (u_{i'} \cdot x_{i'})} \geq \alpha - 2.
    \]
    But since $\alpha > 2$, this implies the derivative is positive, i.e.\ for small enough
    $\epsilon$ the product of the agents' utilities is increasing.
    This contradicts the fact that $x$ is an optimum solution to \eqref{cp:nash_one_sided}.
\end{proof}

We remark that this bound is tight since \cite{AB20} give an instance in which an agent
envies another agent by a factor of 2. See Figure~\ref{fig:nash_2_envy}.

\begin{figure}[hbt]
    \centering
    \begin{tikzpicture}
        \node[circle,fill,inner sep=1.5pt] (v1) at (0, 0) {};
        \node[circle,fill,inner sep=1.5pt] (v2) at (0, 2) {};

        \node[circle,fill,inner sep=1.5pt] (w1) at (4, 0) {};
        \node[circle,fill,inner sep=1.5pt] (w2) at (4, 2) {};

        \node[left] at (v1) {$i'$};
        \node[left] at (v2) {$i$};

        \node[right] at (w1) {$j'$};
        \node[right] at (w2) {$j$};

        \draw[-, dashed] (v1) -- (w1);
        \draw[-] (v1) -- (w2);

        \draw[-, dashed] (v2) -- (w2);
    \end{tikzpicture}
    \caption{Shown is an example instance which demonstrates that 2-EF is tight for Nash bargaining.
    Dashed edges have utility 1, solid edges have utility 2, and missing edges have utility 0.
    Clearly both agents prefer $j$ to $j'$. A simple calculation shows that in the Nash bargaining
    solution, $i$ will get all of $j$ and thus $i'$ will envy $i$ by a factor of
    2.\label{fig:nash_2_envy}}
\end{figure}
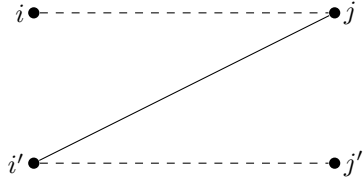

\begin{definition}
    Consider some mechanism $M$ which maps utility profiles $(u_i)_{i \in A}$ to allocations
    $(x_i)_{i \in A}$.
    Then $M$ is called \emph{$\alpha$-incentive compatible} or just $\alpha$-IC if, whenever
    utilities $u$ and $\hat{u}$ differ only on agent $i$, said agent does not improve by more than a
    factor of $\alpha$ wrt.\ to utilities $u$, i.e.\ 
    $u_i \cdot M(u)_i \geq \frac{1}{\alpha} u_i \cdot M(\hat{u})_i$.
    This means that no agent stands to gain more than a factor of $\alpha$ by misreporting their
    utilities.
\end{definition}

\begin{theorem}
    Any mechanism which maps $u$ to some maximizer of \eqref{cp:nash_one_sided} is 2-IC.
\end{theorem}

\begin{proof}
    The proof of this result is quite similar to the proof of Theorem~\ref{thm:nash_2_ef}.
    Consider the original utility profile $u$ and a modified utility profile $\hat{u}$ which differs
    only on one agent, say agent $l \in A$.
    Let $x$ be a maximizer of \eqref{cp:nash_one_sided} under utilities $u$ and $y$ a maximizer of
    \eqref{cp:nash_one_sided} under utilities $\hat{u}$.
    Assume that $u_l \cdot y_l = \alpha u_l \cdot x_l$.
    Our goal is to show that $\alpha \leq 2$.

    For small $\epsilon$, we now consider the new allocations $(1 - \epsilon) x + \epsilon y$.
    This allocation cannot increase the product of the utilities $\hat{u}$ compared to $x$ by the
    maximality of $x$.
    Thus the derivative wrt.\ to $\epsilon$ of
    \[
        \prod_{i \in A} (u_i \cdot x_i (1 - \epsilon) + u_i \cdot y_i \epsilon)
    \]
    must be non-positive at $\epsilon = 0$.
    Performing this computation yields
    \[
        \sum_{i \in A} \frac{u_i \cdot y_i - u_i \cdot x_i}{u_i \cdot x_i} \prod_{i'
        \in A} u_i \cdot x_i \leq 0
    \]
    and therefore
    \[
        \sum_{i \in A \setminus \{l\}} \left(\frac{u_i \cdot y_i}{u_i \cdot x_i} - 1\right) \leq 1 -
        \frac{u_l y_l}{u_l x_l} = 1 - \alpha.
    \]

    The same argument applies to the allocation $\epsilon x + (1 - \epsilon) y$ and the utilities
    $\hat{u}$ by symmetry, giving the inequality
    \[
        \sum_{i \in A \setminus \{l\}} \left(\frac{\hat{u}_i \cdot x_i}{\hat{u}_i \cdot y_i} - 1\right) \leq 1 -
        \frac{\hat{u}_l x_l}{\hat{u}_l y_l} \leq 1.
    \]

    Finally note that for all $i \in A \setminus \{l\}$ we have that $u_i = \hat{u_i}$ so after
    summing up the two inequalities we get:
    \[
        \sum_{i \in A \setminus \{j\}} \left(\frac{u_i \cdot y_i}{u_i \cdot x_i} + \frac{u_i \cdot
        x_i}{u_i \cdot y_i} - 2\right) \leq 2 - \alpha.
    \]
    By the AM-GM inequality, we know that $\frac{a}{b} + \frac{b}{a} \geq 2$ for all $a, b > 0$ and
    so this implies that $2 - \alpha \geq 0$ which is precisely what we wanted to show.
\end{proof}

This bound is also tight as shown by the following family of instances. See
Figure~\ref{fig:nash_2_ic} for an example.

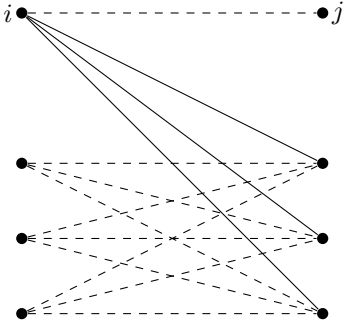
\begin{figure}[hbt]
    \centering
    \begin{tikzpicture}
        \node[circle,fill,inner sep=1.5pt] (v1) at (0, 0) {};
        \node[circle,fill,inner sep=1.5pt] (v2) at (0, 1) {};
        \node[circle,fill,inner sep=1.5pt] (v3) at (0, 2) {};
        \node[circle,fill,inner sep=1.5pt] (v4) at (0, 4) {};

        \node[circle,fill,inner sep=1.5pt] (w1) at (4, 0) {};
        \node[circle,fill,inner sep=1.5pt] (w2) at (4, 1) {};
        \node[circle,fill,inner sep=1.5pt] (w3) at (4, 2) {};
        \node[circle,fill,inner sep=1.5pt] (w4) at (4, 4) {};

        \draw[-, dashed] (v1) -- (w1);
        \draw[-, dashed] (v1) -- (w2);
        \draw[-, dashed] (v1) -- (w3);

        \draw[-, dashed] (v2) -- (w1);
        \draw[-, dashed] (v2) -- (w2);
        \draw[-, dashed] (v2) -- (w3);

        \draw[-, dashed] (v3) -- (w1);
        \draw[-, dashed] (v3) -- (w2);
        \draw[-, dashed] (v3) -- (w3);

        \draw[-, dashed] (v4) -- (w4);
        \draw[-] (v4) -- (w3);
        \draw[-] (v4) -- (w2);
        \draw[-] (v4) -- (w1);

        \node[left] at (v4) {$i$};
        \node[right] at (w4) {$j$};
        %
        %
        %
    \end{tikzpicture}
    \caption{Shown is an example instance from the proof of Theorem~\ref{thm:nash_2_ic} with $n =
    4$. Dashed edges have utility 1, solid edges utility 2, and missing edges have utility 0. Agent
    $i$ will be fully allocated to good $j$ by Nash bargaining even though they would prefer the
    ``desirable'' goods. However, agent $i$ can misrepresent their utilities to look like the
    other agents therefore get a significant fraction of the desirable goods. \label{fig:nash_2_ic}}
\end{figure}

\begin{theorem}\label{thm:nash_2_ic}
    Any mechanism which maps $u$ to some maximizer of \eqref{cp:nash_one_sided} is not $(2 -
    \epsilon)$-IC for any $\epsilon > 0$.
\end{theorem}

\begin{proof}
    Consider the following instance with $n$ agents and $n$ goods.
    Let there be $n - 1$ desirable goods and one undesirable good.
    Agent 1 (the agent who will be incentivized to lie) has utility 2 for the desirable goods and
    utility 1 for the undesirable good whereas all other agents have utility 1 for the desirable
    goods and utility 0 for the undesirable good.
    See Figure~\ref{fig:nash_2_ic} for $n = 4$.

    Let $x$ be the amount that agent 1 is matched to the desirable goods.
    By symmetry\footnote{One can easily see that in optimum solutions to \eqref{cp:nash_one_sided},
    agents with equal utility vectors must have the same overall utility. Otherwise the product of
    their utilities can be improved.}, all other agents must be matched $\frac{n - 1 - x}{n - 1} = 1 - \frac{x}{n - 1}$ to the desirable goods.
    The product of agents' utilities is therefore
    \[
        (2x + (1-x)) \left(1 - \frac{x}{n - 1}\right)^{n - 1}
    \]
    and one may check that this is uniquely maximized at $x = 0$.
    In other words, agent 1 gets nothing from the desirable goods and their utility is 1.

    Now agent 1 misreports their utilities as having utility 1 for the desirable good and utility 0
    for the undesirable good, i.e.\ they report the same utilities as all the other agents.
    But then, by symmetry, this means that agent 1 now gets an equal amount of the desirable goods
    as all the other agents, i.e.\ they get $\frac{n - 1}{n}$ desirable goods.
    Thus their actual utility is $2 \frac{n - 1}{n} + \frac{1}{n}$.
    Finally, as $n \rightarrow \infty$, this implies that any mechanism based on Nash-bargaining
    cannot be better than 2-IC.
\end{proof}

\cite{PTV21} give simple, practical algorithms for computing $(1 +
\epsilon)$-approximate\footnote{Approximate in the sense that all utilities are within $(1 +
\epsilon)$ of the Nash bargaining point.} Nash bargaining points in $O(\poly(n, 1 / \epsilon))$ time
and so we get the following corollary.

\begin{corollary}
    There is a $(2 + \epsilon)$-EF, PO, $(2 + \epsilon)$-IC mechanism for one-sided cardinal-utility
    matching markets which runs in $O(\poly(n, 1 / \epsilon))$ time.
\end{corollary}

Finally, note that \cite{HV21} and \cite{PTV21} also deal with more general settings in which the
agents' utilities are not necessarily linear but given by more general (piecewise-linear) concave
functions.
The above proofs can be adapted to work for non-linear concave utility functions as well, though
this is beyond the scope of this paper.

\section{Two-Sided Matching Markets}

A second interesting class of matching markets is that of two-sided markets in which instead of
matching goods to agents we match agents to other agents.
These markets can be distinguished based on two criteria: whether the underlying graph is bipartite
or not and whether the agents' utilities are symmetric or asymmetric.

In a bipartite matching market, we have two sets $A, B$ of agents with $|A| = |B| = n$ and our goal
is to match each agent in $A$ to an agent in $B$.
Every $i \in A$ has non-negative utilities $u_{i j}$ over $j \in B$ and, likewise, every $j \in
B$ has non-negative utilities $w_{j i}$ over $i \in A$.
A classic example  of this is school choice: students have preferences over schools and schools have
preferences over students (e.g.\ based on test scores).
By a slight abuse of notation we use $w_j \cdot x_j$ to mean $\sum_{i \in A} w_{j i} x_{i j}$.

As mentioned in the introduction, there are also non-bipartite matching markets in which we are
simply given a set of $2 n$ agents and each agent may have utilities over all other agents.
In this case one has to be careful with allowing fractional allocations since fractional
perfect matchings cannot always be decomposed into integral perfect matchings.
Still, these markets are a direct generalization of the bipartite case and so our negative results
apply to them as well.
In the remainder of this section, we will only consider bipartite two-sided matching markets.

The definitions of Pareto-optimality and envy-freeness extend naturally to this setting.

\begin{definition}
    An allocation in a two-sided matching market is \emph{Pareto-optimal} if there is no way to
    increase the utility of \emph{any} agent (on either side) without
    decreasing the utility of another agent (on either side).
\end{definition}

\begin{definition}
    An allocation in a two-sided matching market is \emph{envy-free} if no agent prefers another
    agent's bundle (on their own side) to their own.
\end{definition}

Lastly, we will say that a two-sided market has \emph{symmetric} utilities if $u_{i j} = w_{j i}$
for all $i \in A, j \in B$.
This is mostly of interest when dealing with $\{0, 1\}$ utilities, in which case a pair of agents is
either considered acceptable or not by both parties.

\cite{BM04} showed that in the case of a symmetric, bipartite two-sided matching market with $\{0,
1\}$ utilities, rational EF+PO allocations exist.
Computability is not directly addressed in their paper, though the algorithm by \cite{VY21} can be
adapted for this setting.
This result was extended to the non-bipartite case by \cite{RSU05} who proved existence and
\cite{LLHT14} who gave a polynomial time algorithm.

\subsection{Rationality}

As was the case for one-sided markets, we can show that if an EF+PO allocation exists, there must be
a rational EF+PO allocation.
The proofs are essentially identical to those in Section~\ref{sec:rationality} so we will not
restate them here.

\begin{lemma}
    $x^\star \in P_\mathrm{PM}$ is Pareto-optimal if and only if there exist positive $(\alpha_i)_{i
    \in A}$ and $(\beta_j)_{j \in B}$ such that $x^\star$ maximizes $\phi(x) \coloneqq \sum_{i \in
    A} \alpha_i u_i \cdot x_i + \sum_{j \in B} \beta_j w_j \cdot x_j$ over all $x \in P_\mathrm{PM}$.
    Moreover, if $x^\star$ is rational, $\alpha$ and $\beta$ can be computed in polynomial time.
\end{lemma}

The set of all envy-free allocations is given by the polytope $P_{\mathrm{2EF}}$:

\begin{equation*}
    P_\mathrm{2EF} \coloneqq \left\{ (x_{i j})_{i \in A, j \in G} \quad \middle| \quad
    \begin{matrix}
        \sum_{j \in G} x_{i j} &= 1 & \forall i \in A, \\
        \sum_{i \in A} x_{i j} &= 1 & \forall j \in G, \\
        u_i \cdot x_i - u_i \cdot x_{i'} &\geq 0 & \forall i, i' \in A, \\
        w_j \cdot x_j - w_j \cdot x_{j'} &\geq 0 & \forall j, j' \in B, \\
        x_{i j} &\geq 0 & \forall i \in A, j \in G.
    \end{matrix}
    \right\}
\end{equation*}

\begin{theorem}\label{thm:two_sided_efpo_exists}
    If an instance of a two-sided bipartite matching market admits an EF+PO allocation, then it
    there is one which is a vertex of $P_\mathrm{2EF}$ and is thus rational.
\end{theorem}

We will also need the following characterization in Section~\ref{sec:justified_envy_freeness}.
Note that an allocation is weak Pareto-optimal if there is no other allocation that improves on the
utility of \emph{every} agent.
For a proof see Appendix~\ref{sec:weak_po_characterization}.

\begin{lemma}\label{lem:weak_po_characterization}
    $x^\star \in P_\mathrm{PM}$ is weak Pareto-optimal if and only if there exist non-negative
    $(\alpha_i)_{i \in A}$ and $(\beta_j)_{j \in B}$ such that $\sum_{i \in A} \alpha_i + \sum_{j
    \in B} \beta_j > 0$ and $x^\star$ maximizes $\phi(x) \coloneqq \sum_{i \in A} \alpha_i u_i \cdot
    x_i + \sum_{j \in B} \beta_j w_j \cdot x_j$ over all $x \in P_\mathrm{PM}$.
\end{lemma}

\subsection{Non-Existence of EF+PO Solutions}\label{sec:non_existence}

Given that we know that rational EF+PO allocations exist in various matching markets, even
two-sided non-bipartite markets with $\{0, 1\}$-utilities, an interesting question is whether such
allocations exist for any larger classes of instances.
We will answer this question in the negative by giving rather limiting counterexamples below.

Per Theorem~\ref{thm:two_sided_efpo_exists}, if an EF+PO allocation exists, then it must be a vertex
of the polytope $P_\mathrm{2EF}$.
Such allocations can often be found heuristically: repeatedly pick random vectors $\alpha \in (0,
1]^A$ and $\beta \in (0, 1]^B$ and maximize $\sum_{i \in A} \alpha_i u_i \cdot x_i + \sum_{j \in B}
\beta_j w_j \cdot x_j$ over $P_{\mathrm{2EF}}$ using an LP solver.
This produces a candidate solution $x$ which is Pareto-optimal \emph{among the envy-free
allocations}.
We can then check whether $x$ is Pareto-optimal \emph{among all solutions} by solving the LP
\begin{maxi*}
    {}
    {\sum_{i \in A} u_i \cdot y_i + \sum_{j \in B} w_j \cdot y_j}
    {}
    {}
    \addConstraint{\sum_{j \in B} y_{i j}}{= 1}{\quad \forall i \in A}
    \addConstraint{\sum_{i \in A} y_{i j}}{= 1}{\quad \forall j \in B}
    \addConstraint{u_i \cdot y_i}{\geq u_i \cdot x_i}{\quad \forall i \in A}
    \addConstraint{w_j \cdot y_j}{\geq w_j \cdot x_j}{\quad \forall j \in B}
    \addConstraint{y_{i j}}{\geq 0}{\quad \forall i \in A, j \in B}.
\end{maxi*}

In most instances, this finds an EF+PO allocation relatively quickly.
By enumerating small instaces we found the examples below which have the fewest positive entries in
their utility matrices.

We remark that given the polyhedral nature of the problem, it is possible to design an exact
algorithm which can determine in finite time whether an instance has an EF+PO allocation and return
it: simply enumerate all vertices of $P_\mathrm{2EF}$ and test each one for Pareto-optimality using
the LP approach mentioned above.
However, this is quite slow in practice due to the exponential number of vertices that
$P_\mathrm{2EF}$ generally has.

\begin{theorem}\label{thm:asymmetric}
    For two-sided matching markets under asymmetric utilities, an EF+PO allocation does not always
    exist, even for the case of $\{0, 1\}$-utilities.
\end{theorem}

\begin{proof}
    Consider the instance shown in Figure~\ref{fig:asymmetric} and the Pareto-optimal fractional
    perfect matching $y$ depicted in that figure.
    Let $x$ be some allocation in this instance and assume that $x$ is envy-free.
    We will show that $y$ is strictly Pareto-better than $x$.

    First let us show that $x_{24} = \frac{1}{3}$.
    Note that we must clearly have $x_{24} \geq \frac{1}{3}$ as otherwise $x_{25} > \frac{1}{3}$ or
    $x_{26} > \frac{1}{3}$ and in those cases agent $4$ would envy agent $5$ or $6$ respectively.
    On the other hand, assume that $x_{24} = \frac{1}{3} + \epsilon$.
    Then $u_2 \cdot x_2 \leq \frac{2}{3} - \epsilon$.
    But then agent 2 envies either agent 1 or agent 3 since among these three, one must get at least
    $\frac{2}{3}$ of agents 5 and 6.
    Thus $x_{24} = \frac{1}{3}$ as claimed.

    Next we claim that $x_{14} = \frac{1}{3}$.
    Again, we clearly have $x_{14} \geq \frac{1}{3}$ as otherwise agent 1 would envy agent 2 or
    agent 3.
    But in the other direciton, if $x_{14} = \frac{1}{3} + \epsilon$, then $x_{15} + x_{16} =
    \frac{2}{3} - \epsilon$.
    By the previous claim, we know that $x_{25} + x_{26} = \frac{2}{3}$ and so $x_{35} + x_{36} =
    \frac{2}{3} + \epsilon$ which would imply that agent 2 envies agent 3.
    Thus $x_{14} = \frac{1}{3}$.

    Finally, since $x_{24} = \frac{1}{3}$ and $x_{14} = \frac{1}{3}$, we can see that $y$ is
    Pareto-better than $x$ (regardless of how $x$ assigns the other edges).
    In particular, it strictly improves the utility of agent 1 and weakly improves the utility of
    all other agents.
\end{proof}

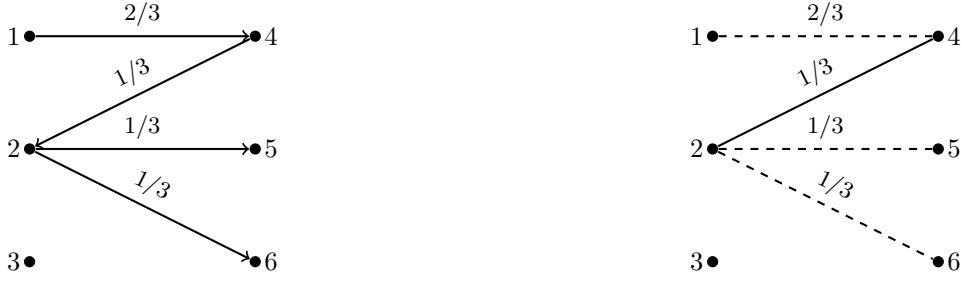
\begin{figure}[htb]
    \centering
    \begin{subfigure}[b]{0.45\textwidth}
        \centering
        \begin{tikzpicture}[scale=1.5]
            \node[circle, fill, inner sep=1.5pt] (v1) at (0, 0) {};
            \node[circle, fill, inner sep=1.5pt] (v2) at (0, -1) {};
            \node[circle, fill, inner sep=1.5pt] (v3) at (0, -2) {};

            \node[circle, fill, inner sep=1.5pt] (w1) at (2, 0) {};
            \node[circle, fill, inner sep=1.5pt] (w2) at (2, -1) {};
            \node[circle, fill, inner sep=1.5pt] (w3) at (2, -2) {};

            \node[left] at (v1) {$1$};
            \node[left] at (v2) {$2$};
            \node[left] at (v3) {$3$};

            \node[right] at (w1) {$4$};
            \node[right] at (w2) {$5$};
            \node[right] at (w3) {$6$};

            \draw[->, thick] (v1) -- node[midway, above, sloped] {\small $2/3$} (w1);
            \draw[->, thick] (w1) -- node[midway, above, sloped] {\small $1/3$} (v2);
            \draw[->, thick] (v2) -- node[midway, above, sloped] {\small $1/3$} (w2);
            \draw[->, thick] (v2) -- node[midway, above, sloped] {\small $1/3$} (w3);
        \end{tikzpicture}
        \caption{Each arrow represents a utility 1-edge from one side (and utility 0 from the other
        side).}
        \label{fig:asymmetric}
    \end{subfigure}
    \hfill
    \begin{subfigure}[b]{0.45\textwidth}
        \centering
        \begin{tikzpicture}[scale=1.5]
            \node[circle, fill, inner sep=1.5pt] (v1) at (0, 0) {};
            \node[circle, fill, inner sep=1.5pt] (v2) at (0, -1) {};
            \node[circle, fill, inner sep=1.5pt] (v3) at (0, -2) {};

            \node[circle, fill, inner sep=1.5pt] (w1) at (2, 0) {};
            \node[circle, fill, inner sep=1.5pt] (w2) at (2, -1) {};
            \node[circle, fill, inner sep=1.5pt] (w3) at (2, -2) {};

            \node[left] at (v1) {$1$};
            \node[left] at (v2) {$2$};
            \node[left] at (v3) {$3$};

            \node[right] at (w1) {$4$};
            \node[right] at (w2) {$5$};
            \node[right] at (w3) {$6$};

            \draw[-, dashed, thick] (v1) -- node[midway, above, sloped] {\small $2/3$} (w1);
            \draw[-, thick] (w1) -- node[midway, above, sloped] {\small $1/3$} (v2);
            \draw[-, dashed, thick] (v2) -- node[midway, above, sloped] {\small $1/3$} (w2);
            \draw[-, dashed, thick] (v2) -- node[midway, above, sloped] {\small $1/3$} (w3);
        \end{tikzpicture}
        \caption{Dashed edges have utility 1, whereas solid edges have utility 2.}
        \label{fig:symmetric}
    \end{subfigure}
    \caption{Shown are counterexamples for $\{0, 1\}$ asymetric (a) and $\{0, 1, 2\}$ symmetric
    utilities. In both cases the edge labels show a Pareto-optimal solution $y$ and edges which are
    not drawn have utility 0 (assume that $y$ is extended to a fractional perfect matching by
    filling up with utility 0 edges).}
\end{figure}

\begin{theorem}\label{thm:sym_counterex}
    For two-sided matching markets under symmetric utilities, an EF+PO allocation does not always
    exist, even in the case of $\{0, 1, 2\}$-utilities.
\end{theorem}

\begin{proof}
    Consider the instance shown in Figure~\ref{fig:symmetric} together with the depicted
    Pareto-optimal allocation $y$.
    Let $x$ be some envy-free allocation.
    We aim to show that $y$ is Pareto-better than $x$.

    First, we can see that $x_{12} = x_{24}$ because if $x_{12} > x_{24}$, then agent 2 would envy
    agent 1 and if $x_{24} > x_{12}$, then agent 1 would envy agent 2.
    Next, observe that $x_{24} \leq \frac{1}{3}$ since otherwise either agent 5 or 6 would envy
    agent 4.
    However, if $x_{24} < \frac{1}{3}$, then the total utility that agent 4 gets (i.e.\ $w_4 \cdot
    x_4$) is strictly less than 1 (since $x_{14} = x_{24}$).
    But since $w_{41} + w_{42} = 3$, at least one agent $j \in \{5, 6\}$ must satisfy $w_4 \cdot x_j
    > 1$ and agent 4 would envy that agent.
    Thus $x_{24} = \frac{1}{3}$.

    Finally, we must have that $x_{25} = x_{26} = \frac{1}{3}$ since otherwise agent 5 would envy
    agent 6 or vice versa.
    This determines $x$ on all the edges with positive utility.
    But now observe that $y$ is Pareto-better than $x$ since it strictly improves the utility of
    agent 1 and weakly improves the utility of all other agents.
\end{proof}

\subsection{Justified Envy-Freeness}\label{sec:justified_envy_freeness}

As we have seen in the previous section, in two-sided markets we generally cannot get EF+PO
allocations unless we are using symmetric $\{0, 1\}$ utilities.
Intuitively, the issue is that agents have different entitlements.
Consider a market in which an agent $i \in A$ is liked by everyone in $B$ whereas $i' \in A$ is
hated by everyone in $B$.
It will be difficult to avoid a situation in which $i$ envies $i'$ without sacrificing efficiency.

A way to get around this is to simply formalize this notion of entitlement.
In the following, fix some bipartite two-sided matching market with $|A| = |B| = n$ and utilities
$u, w$.

\begin{definition}
    In an allocation $x$, agent $i \in A$ has \emph{strong justified envy} towards $i' \in A$ if
    $w_{j i} \geq w_{j i'}$ for all $j \in B$ and $u_i \cdot x_i < u_i \cdot x_{i'}$.
    Strong justified envy is defined symmetrically for agents in $B$.
    An allocation in which there is no strong justified envy is said to be \emph{weak justified
    envy-free (weak JEF)}.
\end{definition}

Weak justified envy-freeness is a reasonable notion in many settings.
For example, in school choice a student who scores higher on all relevant tests should not envy a
student who scores lower.
However, it is somewhat unsatisfying that $i$ is only justified in their envy of $i'$ when
\emph{all} agents prefer $i$ to $i'$, even agents that $i$ does not care about.
For this reason, we define a stronger notion of justified envy-freeness.

\begin{definition}
    In an allocation $x$, agent $i \in A$ has \emph{justified envy} towards $i' \in A$ if
    \[
        u_i \cdot x_i < \sum_{\substack{j \in B\\w_{j i} \geq w_{j i'}}} u_{i j} x_{i' j}
    \]
    and likewise for agents in $B$.
    An allocation in which there is no justified envy is \emph{justified envy-free (JEF)}.
\end{definition}

Clearly, strong justified envy implies justified envy and therefore JEF implies weak JEF.
We remark that in the case of an integral matching, being JEF is equivalent to being a stable
matching.
We will show the following.

\begin{theorem}\label{thm:jef_wpo_exists}
    There always exists a rational allocation which is JEF and weak PO.
\end{theorem}

The proof uses a limit argument based on an equilibrium notion introduced by \cite{M16}.
This equilibrium is conceptually similar to an HZ equilibrium with three crucial differences: 
\begin{itemize}
    \item While each agent is endowed with some amount of fake currency, the value of this currency
        is not normalized.
        Instead there is a price $p_m$ that determines the ``price of money''.
    \item Prices are double-indexed, i.e.\ an agent in $B$ may have different prices for every agent
        in $A$.
    \item Prices can be negative. They effectively represent transfers between the two sides of
        agents.
\end{itemize}

We do not need the full generality of the equilibrium notion of Manjunath and will give a slightly
simplified definition assuming linear utilities.
Each agent $i \in A$ (and likewise for agents in $B$) has some initial endowment $\omega_i > 0$ of
``money'' and they will receive not just an allocation $(x_{i j})_{j \in B}$ but also some money
$m_i \geq 0$.
We assume that their utility is given by $u_i \cdot x_i + m_i$.
Likewise for the agents in $B$.

\begin{definition}
    An \emph{double-indexed price (DIP) equilibrium} consists of an assignment $(x_{i j})_{i \in A,
    j \in B}$, money assignments $(m_k)_{k \in A \cup B}$, individualized prices $(p_{i j})_{i \in
    A, j \in B}$ and $(q_{j i})_{j \in B, i \in A}$, and the price of money $p_m$ satisfying:
    \begin{enumerate}
        \item $x$ is a fractional matching (but not necessarily perfect).
        \item The money is redistributed exactly, i.e. $\sum_{i \in A \cup B} \omega_i = \sum_{i \in
            A \cup B} m_i$.
        \item Each agent $i \in A$ (and likewise for agents in $B$) receives an optimal bundle in
            the sense that $(x_i, m_i)$ maximizes
            \begin{maxi*}
                {}
                {u_i \cdot x_i + m_i}
                {}
                {}
                \addConstraint{\sum_{j \in B} x_{i j}}{\leq 1}
                \addConstraint{p_i \cdot x_i + p_m m_i}{\leq p_m \omega_i}
                \addConstraint{x_{i j}}{\geq 0}{\quad \forall j \in B.}
            \end{maxi*}
        \item $p_{i j} = -q_{j i}$ for all $i \in A, j \in B$.
    \end{enumerate}
\end{definition}

\begin{theorem}[\cite{M16}]\label{thm:dip_exists}
    As long as every agent has a positive endowment of money (i.e.\ $\omega_i > 0$), a DIP
    equilibrium always exists.
\end{theorem}

\begin{theorem}[\cite{M16}]\label{thm:dip_po}
    The allocation in a DIP equilibrium is Pareto-optimal.
\end{theorem}

We require the allocation to be a fractional perfect matching.
It is possible to modify the proof of Theorem~\ref{thm:dip_exists} directly but in order to be
self-contained, we will give a short proof which uses Theorem~\ref{thm:dip_exists} as a black box.

\begin{lemma}
    As long as every agent has a positive endowment of money (i.e.\ $\omega_i > 0$), a DIP
    equilibrium in which $x$ is a fractional perfect matching always exists.
\end{lemma}

\begin{proof}
    For each $k \in \bbN^+$, consider a modified instance in which every zero utility is replaced by
    $\frac{1}{k}$.
    Each of these instances has some DIP equilibrium and clearly in each of these equilibria, the
    allocation must be a fractional perfect matching since otherwise this would immediately violate
    Pareto-optimality.

    Since the prices are scale invariant, we can rescale them so that the maximum price is bounded
    by 1.
    Then both allocations, money assignments, and prices are bounded so by compactness one can find
    a convergent subsequence of these DIP equilibria.
    The limiting point is a DIP equilibrium in the original instance with an allocation which is a
    fractional perfect matching.
\end{proof}

\begin{lemma}\label{lem:eps_jef_exists}
    If $\omega_i = \frac{\epsilon}{2 n}$ for all $i \in A \cup B$, and $(x, m, p, q, p_m)$ is a DIP
    equilibrium for these budgets, then for all $i, i' \in A$ (and likewise for agents in $B$) we
    have
    \[
        u_i \cdot x_i \geq \sum_{\substack{j \in B\\ w_{j i} \geq w_{j i'}}}{u_{i j} x_{i' j}} -
        \epsilon.
    \]
\end{lemma}

\begin{proof}
    Let $i, i' \in A$.
    Consider $j \in B$ with $x_{i' j} > 0$ and $w_{j i} \geq w_{j i'}$.
    Then we can see that $p_{i j} \leq p_{i' j}$.
    If this were not the case, then since $q_{j i} = -p_{i j}$ and $q_{j i'} = -p_{i' j}$, we
    would have $q_{j i} < q_{j i'}$ and thus $j$ could redistribute some of their bundle from $i'$
    to $i$ decreasing their total expenditure withiout decreasing their utility.
    This is a contradiction to the fact that $j$ gets an optimal bundle since they could then
    increase $m_j$ to get a strictly better bundle.

    This means that
    \[
        \sum_{\substack{j \in B\\ w_{j i} \geq w_{j i'}}}{p_{i j} x_{i' j}} \leq
        \sum_{\substack{j \in B\\ w_{j i} \geq w_{j i'}}}{p_{i' j} x_{i' j}} \leq p_m (\omega_{i'} -
        m_{i'}) \leq p_m \omega_{i'} = p_m \omega_i
    \]
    where we used that all agents have equal endowments of money in the last equality.
    But since $i$ maximizes their utility among all bundles which cost at most $p_m \omega_i$, this
    implies that
    \[
        u_i \cdot x_i + m_i \geq 
        \sum_{\substack{j \in B\\ w_{j i} \geq w_{j i'}}}{u_{i j} x_{i' j}}.
    \]

    Finally note that $m_i \leq \sum_{k \in A \cup B} m_k = \sum_{k \in A \cup B} \omega_k =
    \epsilon$ and this finishes the proof.
    By symmetry the same holds for all pairs of agents in $B$.
\end{proof}

\begin{proof}[Proof of Theorem~\ref{thm:jef_wpo_exists}]
    First, let us show that a JEF and weak PO allocation exists.
    By Lemma~\ref{lem:eps_jef_exists}, we can pick a sequence $x^{(k)}$ of Pareto-optimal
    allocations such that
    \[
        u_i \cdot x^{(k)}_i \geq \sum_{\substack{j \in B\\ w_{j i} \geq w_{j i'}}}{u_{i j}
        x^{(k)}_{i' j}} - \epsilon_k
    \]
    for all $i, i' \in A$ (likewise for agents in $B$) and $\epsilon_k \rightarrow 0$.
    Since the set of all fractional perfect matchings is compact, we can find a convergent
    subsequence.
    Without loss of generality, assume that $x^{(k)}$ converges to some $x^\star$.
    Clearly $x^\star$ is itself a fractional perfect matching.

    The limit point of a sequence of Pareto-optimal allocations is a weak Pareto-optimal allocation.
    Furthermore, it is easy to see that $x^\star$ is justified envy-free.

    Lastly, we can use a similar argument as in the proof of Theorem~\ref{thm:efpo_rational} to show
    that a rational JEF + weak PO allocation exists as well.
    Simply pick $\alpha, \beta$ according to Lemma~\ref{lem:weak_po_characterization} and then find
    a vertex solution which maximizes $\sum_{i \in A} \alpha_i u_i \cdot x_i + \sum_{j \in B} \beta_j
    w_j \cdot x_j$ over the polytope of all justified envy-free allocations.
\end{proof}

\subsection{Justified Envy for Nash Bargaining}\label{sec:two_sided_nash}

As shown in Section~\ref{sec:one_sided_nash}, Nash bargaining yields an approximately envy-free and
Pareto-optimal allocation in the case of one-sided matching markets.
One might reasonably conjecture that it achieves approximately \emph{justified} envy-freeness in the
two-sided setting.
We give a counterexample below based on a similar example due to \cite{PTV21} that shows this not to
be the case.

\begin{theorem}
    There are instances on $n$ vertices such that in the Nash bargaining solution $x$, there are
    agents $i, i' \in A$ such that all agents in $B$ prefer $i$ to $i'$ and yet $u_i \cdot x_i =
    \frac{1}{n} u_i \cdot x_{i'}$.
\end{theorem}

\begin{proof}
    Our instance has three special agents: $i, i' \in A$ and $j \in B$.
    All agents in $B \setminus \{j\}$ have utility 1 for $i$ but 0 for everyone else in $A$,
    including $i'$.
    Agents $i$ and $i'$ both have utility 1 for agent $j$ and utility 0 for all other agents.
    The agents in $A \setminus \{i, i'\}$ are dummy agents and have identical utility for all agents
    in $B$.
    See Figure~\ref{fig:nash_envy}.

    Consider a Nash bargaining solution $x$.
    The agents in $B \setminus \{j\}$ must all be allocated an equal amount of agent $i$, since
    otherwise we could increase the product of the agents' utilities by making them equal.
    Let $y$ be this amount.
    Then we must have $x_{i j} = 1 - (n - 1) y$ and $x_{i' j} = (n - 1) y$.
    Therefore $y$ must maximize
    \[
        (1 - (n - 1) y) \cdot (n - 1) y \cdot y^{n - 1}
    \]
    which implies that $y = \frac{n}{n^2 - 1}$.
    Then we can compute that $u_i \cdot x_i = \frac{1}{n + 1}$ whereas $u_i \cdot x_{i'} =
    \frac{n}{n + 1}$.
\end{proof}

\begin{figure}[htb]
    \centering
    \begin{tikzpicture}[xscale=2, yscale=0.5,
                        every node/.style={circle, inner sep=1.5pt},
                        grey node/.style={fill=gray},
                        black node/.style={fill}]
        \foreach \y in {1,...,6} {
            \node[grey node] (L\y) at (0,1.5*\y) {};
        }
        \foreach \y in {7,...,8} {
            \node[black node] (L\y) at (0,1.5*\y) {};
        }

        \foreach \y in {1,...,8} {
            \node[black node] (R\y) at (3,1.5*\y) {};
        }

        \foreach \y in {1,...,7} {
            \draw[->] (R\y) -- (L7);
        }

        \draw[->] (L8) -- (R8);
        \draw[<->] (L7) -- (R8);

        \node[left] at (L7) {$i$};
        \node[left] at (L8) {$i'$};
        \node[right] at (R8) {$j$};
    \end{tikzpicture}
    \caption{Shown is an instance ($n = 8$) with strong justified envy for Nash bargaining. Agents
    $i$ and $i'$ compete for $j$ but all agents in $B \setminus \{j\}$ want $i$ so $i$ gets only a
    small fraction of $j$. The gray agents are dummy agents and have identical utilities for all
    agents in $B$.
    \label{fig:nash_envy}}
\end{figure}
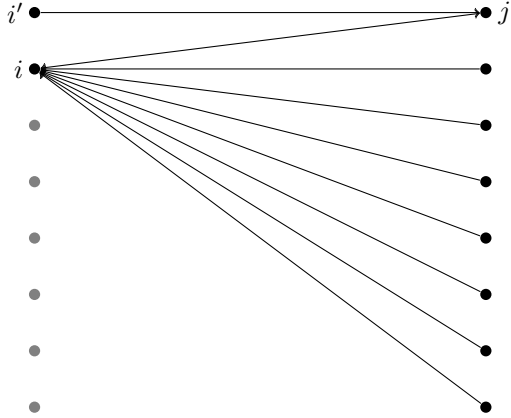

\section{Conclusion}

We have resolved the question of whether we can obtain polynomial time mechanisms which give EF+PO
lotteries in cardinal-utility matching markets: we can not unless $\mathrm{FP} = \mathrm{PPAD}$.
However, this leaves several interesting open questions:
\begin{itemize}
    \item Is there a polynomial time algorithm to find $\alpha$-approximately JEF+PO lotteries in
        two-sided markets, for any constant $\alpha$?
    \item Is Nash bargaining the best we can do for one-sided markets or is there a way to compute
        an $\alpha$-envy-free and Pareto-optimal lottery for $\alpha < 2$ in polynomial time?
    \item Is there a way to compute an envy-free lottery in polynomial time which satisfies some
        relaxed notion of Pareto-optimality?
\end{itemize}

\section*{Acknowledgements}

The authors would like to thank Bernard Salani\'e for mentioning continuum markets to us.
We would also like to thank Binghui Peng, Federico Echenique, and Joseph Root for valuable
discussions.
Finally, we would also like to thank Ben Cookson and Nisarg Shah for spotting an error in our
proof of Theorem~\ref{thm:sym_counterex}.

\bibliographystyle{plainnat}
\bibliography{references}

\newpage
\appendix

\section{Characterization of Pareto-Optimality}\label{sec:pareto_characterization}

\begin{proof}[Proof of Lemma~\ref{lem:pareto_characterization}]
    Clearly if $x^\star$ maximizes $\phi(x)$, then it is a Pareto-optimal allocation since any
    Pareto-better allocation $x$ would satisfy $\phi(x) > \phi(x^\star)$ since $\alpha$ is strictly
    positive.

    For the other direction, note that by Pareto-optimality, $x^\star$ is a maximizer of the linear
    program:
    \begin{maxi*}
        {}
        {\sum_{i \in A} u_i \cdot x_i}
        {}
        {}
        \addConstraint{u_i \cdot x_i}{\geq u_i \cdot x^\star_i}{\quad \forall i \in A}
        \addConstraint{\sum_{j \in G} x_{i j}}{= 1}{\quad \forall i \in A}
        \addConstraint{\sum_{i \in A} x_{i j}}{= 1}{\quad \forall j \in G}
        \addConstraint{x_{i j}}{\geq 0}{\quad \forall i \in A, j \in G.}
    \end{maxi*}

    Consider a solution $(a, q, p)$ to the dual program:
    \begin{minie}
        {}
        {\sum_{i \in A}{a_i u_i \cdot x^\star_i} + \sum_{i \in A} q_i + \sum_{j \in G} p_j}
        {}
        {}
        \addConstraint{a_i u_{i j} + q_i + p_j}{\geq u_{i j}}{\quad \forall i \in A, j \in
        G\label{eq:pareto_feasibility}}
        \addConstraint{a_i}{\leq 0}{\quad \forall i \in A}.
    \end{minie}
    Then by strong duality
    \begin{equation}
        \sum_{i \in A} u_i \cdot x^\star_i = \sum_{i \in A} a_i u_i \cdot x^\star_i + \sum_{i \in A} q_i +
        \sum_{j \in G} p_j.\label{eq:pareto_strong_duality}
    \end{equation}

    Define $\alpha_i \coloneqq 1 - a_i$.
    Then clearly $\alpha_i > 0$ for all $i$ since $a_i \leq 0$.
    Now we want to show that $x^\star$ is a maximizer of
    \begin{maxi*}
        {}
        {\sum_{i \in A} \alpha_i u_i \cdot x_i}
        {}
        {}
        \addConstraint{\sum_{i \in G} x_{i j}}{= 1}{\quad \forall i \in A}
        \addConstraint{\sum_{j \in A} x_{i j}}{= 1}{\quad \forall j \in G}
        \addConstraint{x_{i j}}{\geq 0}{\quad \forall i \in A, j \in G.}
    \end{maxi*}
    But this follows immediately from the fact that $(q, p)$ is an optimal dual solution to this LP:
    \eqref{eq:pareto_feasibility} implies feasibility and \eqref{eq:pareto_strong_duality} implies
    optimality.
\end{proof}

\section{Characterization of Weak Pareto-Optimality}\label{sec:weak_po_characterization}

\begin{proof}[Proof of Lemma~\ref{lem:weak_po_characterization}]
    The proof is quite similar to the proof of Lemma~\ref{lem:pareto_characterization}.
    Clearly if $x^\star$ maximizes $\phi(x)$, then it is a Pareto-optimal allocation since any
    strong Pareto-better allocation $x$ would satisfy $\phi(x) > \phi(x^\star)$ since at least one
    $\alpha_i$ or $\beta_j$ is positive.

    For the other direction, note that by weak Pareto-optimality, $(x^\star, 0)$ is a maximizer of
    the linear program:
    \begin{maxi*}
        {}
        {t}
        {}
        {}
        \addConstraint{u_i \cdot x_i - t}{\geq u_i \cdot x^\star_i}{\quad \forall i \in A}
        \addConstraint{w_j \cdot x_j - t}{\geq w_j \cdot x^\star_j}{\quad \forall j \in B}
        \addConstraint{\sum_{j \in B} x_{i j}}{= 1}{\quad \forall i \in A}
        \addConstraint{\sum_{i \in A} x_{i j}}{= 1}{\quad \forall j \in B}
        \addConstraint{x_{i j}}{\geq 0}{\quad \forall i \in A, j \in B}
        \addConstraint{t}{\geq 0.}
    \end{maxi*}

    Consider an optimum solution $(\alpha, \beta, p, q)$ to the dual program:
    \begin{mini*}
        {}
        {\sum_{i \in A} q_i + \sum_{j \in B} p_j - \sum_{i \in A} \alpha_i u_i \cdot x_i - \sum_{j \
        B} \beta_j w_j \cdot x_j}
        {}
        {}
        \addConstraint{q_i + p_j - \alpha_i u_{i j} - \beta_j w_{j i}}{\geq 0}{\quad \forall i \in
        A, j \in B}
        \addConstraint{\sum_{i \in A} \alpha_i + \sum_{j \in B} \beta_j}{\geq 1}
        \addConstraint{\alpha_i}{\geq 0}{\quad \forall i \in A}
        \addConstraint{\beta_i}{\geq 0}{\quad \forall j \in B}
    \end{mini*}
    Then we may see that $x^\star$ is an optimum solution to
    \begin{maxi*}
        {}
        {\sum_{i \in A} \alpha_i u_i \cdot x_i + \sum_{j \in B} \beta_j w_j \cdot x_j}
        {}
        {}
        \addConstraint{\sum_{i \in B} x_{i j}}{= 1}{\quad \forall i \in A}
        \addConstraint{\sum_{j \in A} x_{i j}}{= 1}{\quad \forall j \in B}
        \addConstraint{x_{i j}}{\geq 0}{\quad \forall i \in A, j \in B.}
    \end{maxi*}
    since $(p, q)$ gives an optimum dual solution.
\end{proof}

\end{document}